\documentclass[11pt]{article}

\usepackage{amsfonts,amsmath, amsthm}
\usepackage{amssymb}
\usepackage{paralist}
\usepackage{multirow}
\usepackage{xspace}
\usepackage{graphicx}
\graphicspath{ {.} }
\usepackage{tikz}
\usepackage[textsize=scriptsize]{todonotes}
\usepackage[noend,ruled]{algorithm2e}
\usepackage{color}
\usepackage{comment}
\usepackage{float}
\usepackage{tikz}
\usepackage{mathtools}
\usetikzlibrary{positioning}
\usepackage{bbm}
\usepackage{booktabs}
\usepackage[]{units}
\usepackage{pgfplots}
\usetikzlibrary{calc,pgfplots.groupplots}

\usepackage[colorlinks,citecolor=green!60!black,linkcolor=red!60!black]{hyperref}

\newcommand{\reals}{{\mathbb{R}}}
\newcommand{\naturals}{{\mathbb{N}}}

\newtheorem{definition}{Definition}

\newtheorem{theorem}{Theorem}
\newtheorem{proposition}[theorem]{Proposition}
\newtheorem{lemma}[theorem]{Lemma}
\newtheorem{corollary}{Corollary}

\newtheorem{claim}{Claim}
\newtheorem{remark}{Remark}
\newtheorem{example}{Example}

\newcommand{\calR}{\mathcal{R}}

\newcommand{\np}{{\mathrm{NP}}}
\newcommand{\conp}{{\mathrm{coNP}}}

\newcommand{\fpt}{{\mathrm{FPT}}}

\newcommand{\wone}{{\mathsf{W[1]}}}
\newcommand{\wtwo}{{\mathsf{W[2]}}}

\newcommand{\p}{{\mathrm{P}}}

\newcommand{\score}[1]{#1\hbox{-}\mathrm{score}}
\newcommand{\pos}{\mathrm{pos}}
\newcommand{\rank}{\mathrm{rank}}
\newcommand{\reppos}{\mathrm{reppos}}

\newcommand{\bordacc}{{\beta\hbox{-}\mathrm{CC}}}

\newcommand{\topc}{{\mathrm{top}}}
\newcommand{\shift}{{\mathrm{shift}}}

\newcommand{\pref}{\ensuremath{\succ}}

\newenvironment{profile}[1]{
\begin{center}
$\begin{array}{#1}
}{ 
\end{array}$
\end{center}
}

\title{Robustness Among Multiwinner Voting Rules\footnote{A preliminary version
 of this article appeared in \emph{Proceedings of the 10th International
 Symposium on Algorithmic Game Theory, SAGT 2017}~\cite{BFKNST17}.}}

\author{
  Robert Bredereck\thanks{Work done in part while Robert Bredereck was at the University of Oxford.}\\
  TU Berlin\\
  Berlin, Germany\\
  \small{\texttt{robert.bredereck@tu-berlin.de}}
  \and
  Piotr Faliszewski\\
  AGH University\\
  Krakow, Poland \\
  \small{\texttt{faliszew@agh.edu.pl}}
  \and
  Andrzej Kaczmarczyk\\
  TU Berlin\\
  Berlin, Germany\\
  \small{\texttt{a.kaczmarczyk@tu-berlin.de}}
  \and
  Rolf Niedermeier\\
  TU Berlin\\
  Berlin, Germany\\
  \small{\texttt{rolf.niedermeier@tu-berlin.de}}
  \and
  Piotr Skowron\thanks{Work done in part while Piotr Skowron was at TU Berlin.}\\
  University of Warsaw\\
  Warsaw, Poland\\
  \small{\texttt{p.k.skowron@gmail.com}}
  \and
  Nimrod Talmon\thanks{Work done in part while Nimrod Talmon was at the Weizmann Institute of Science.}\\
  Ben-Gurion Univeristy\\
  Be'er Sheva, Israel\\
  \small{\texttt{talmonn@bgu.ac.il}}
}

\begin{document}

\sloppy

\maketitle
\thispagestyle{plain}

\begin{abstract}
%
  We investigate how robust the results of committee elections are to small
  changes in the input preference orders, depending on the voting
  rules used. 
  We find that for typical rules the effect of making a single swap of
  adjacent candidates in a single preference order is either that (1)
  at most one committee member might be replaced, or (2) it is possible
  that the whole committee will be replaced.
%
%
  We also show that the problem of computing the smallest number of
  swaps that lead to changing the election outcome is typically
  $\np$-hard, but there are natural $\fpt$ algorithms. Finally, for a
  number of rules we assess experimentally the average number of random
  swaps necessary to change the election result.

\end{abstract}

\section{Introduction}

We study how multiwinner voting rules---that is, procedures used to select
fixed-size committees of candidates---react to (small) changes in the
input votes.
We are interested both in the complexity of computing the smallest
modification of the votes that affects the election outcome and in the
extent of the possible changes.
We start by discussing our ideas informally in the following example.

Consider a research-funding agency that needs to choose which of the
submitted project proposals to support.
The agency asks a group of experts to evaluate the proposals and to
rank them from the best to the worst one. Then, the agency uses some
formal process---here modeled as a multiwinner voting rule---to
aggregate these rankings and to select $k$ projects 
to be funded.
Let us imagine that one of the experts realized that, instead
of ranking some proposal~$A$ as better than~$B$, he or she should have
given the opposite opinion.
What are the consequences of such a ``mistake'' of the expert? 
It may not affect the results at all, or it may cause
only a 
minor change: 
Perhaps proposal~$A$ would be dropped (to the
benefit of~$B$ or some other proposal) or $B$ would be selected (at
the expense of~$A$ or some other proposal).
We show that, while this indeed would be the case under a number of
multiwinner voting rules (e.g.,\,under the $k$-Borda rule; see
Section~\ref{sec:prelim} for the definitions), there exist
other rules (e.g.,\,Single Transferable Vote, further referred to as STV, or the
Chamberlin--Courant rule) for which such a single swap could lead to selecting a
completely disjoint set of proposals. As the agency would
prefer to avoid situations where small changes in the 
experts'
opinions lead to (possibly large) changes in the outcomes,
the agency would want to be able to compute the smallest number of swaps that
would change the result. In cases where this number
is 
too small, the agency might invite more experts to gain 
confidence in the results.

Below we provide a slightly more formal introduction.
First, a multiwinner voting rule is a function that, given a
set of rankings of the candidates and an integer $k$, outputs a
family of size-$k$ subsets of the candidates (the winning
committees). We consider the following three issues (for simplicity,
below we ignore ties and assume to always have a unique winning committee):



\begin{enumerate}

\item We say that a multiwinner rule $\calR$ is \emph{$\ell$-robust}
  if (1) swapping two adjacent candidates in a single vote can lead to
  replacing no more than $\ell$ candidates in the winning
  committee,\footnote{The formal definition is more complex due to the possibility of ties.} and (2) there are examples where exactly
  $\ell$ candidates are indeed replaced; we refer to $\ell$ as the
  \emph{robustness level} of $\calR$. The robustness level is
  between $1$ and $k$, with $1$-robustness being the strongest form of
  robustness one could ask for. We investigate the robustness levels of
  several multiwinner rules.

\item We say that the \emph{robustness radius} of an election $E$ (for
  committee size~$k$) under a multiwinner rule $\calR$ is the smallest
  number of swaps of adjacent candidates which are necessary to change
  the election outcome. We ask for the complexity of computing the
  robustness radius (referred to as the \textsc{Robustness Radius}
  problem) under a number of multiwinner rules. This problem is
  strongly related to the \textsc{Margin of
    Victory}~\cite{mag-riv-she-wag:c:stv-bribery,car:c:margin-of-victory,xia:margin-of-victory,blom2017towards}
  and \textsc{Destructive Swap Bribery}
  problems~\cite{elk-fal-sli:c:swap-bribery,shi-yu-elk:c:robustness}.
  Furthermore, our work follows up on the study of Shiryaev et
  al.~\cite{shi-yu-elk:c:robustness}, who considered the robustness of
  single-winner rules.

\item In addition to the above-described contributions, we ask how
  many random swaps of adjacent candidates are necessary, on
  average, 
  to move from a randomly generated election
  to one with a different outcome. 
  We assess this kind of robustness of our rules experimentally.

\end{enumerate}

There is quite a number of multiwinner rules. We consider only
several of them, selected to represent a varied set of ideas from the
literature, ranging from variants of scoring rules, through rules
inspired by the Condorcet criterion, to the elimination-based STV
rule.
We find that all these rules 
are either $1$-robust---so a single swap can replace at most one
committee member---or are $k$-robust---so a single swap can replace the
whole committee of size~$k$.\footnote{We also construct somewhat
  artificial rules with robustness levels between $1$ and $k$.} 
Somewhat surprisingly, 
this phenomenon is deeply connected to the complexity of winner
determination. Specifically, under mild assumptions we show that if a
rule has a constant robustness level, then it has a polynomial-time
computable refinement (that is, it is possible to compute \emph{one}
of its outcomes in polynomial time). Since for many rules the problem
of computing such a refinement is $\np$-hard, we get a quick way of
finding out that such rules have nonconstant robustness levels.

%
%
%

The \textsc{Robustness Radius} problem tends to be $\np$-hard
(sometimes even for a single swap) and, thus,
we seek fixed-parameter tractability (FPT)
results. 
For example, we find several FPT algorithms parameterized by the number
of voters (these algorithms are useful, e.g., for scenarios with few
experts, 
such as in our introductory example). See Table~\ref{tab:results} for
an overview of our theoretical results.  We mention that Misra and
Sonar~\cite{mis-son:c:robustness} followed up on our results and, in
particular, have considered several variants of the
Chamberlin--Courant rule and certain nearly-structured preference
domains. Recently, Gawron and Faliszewski~\cite{gaw-fal:c:robustness} applied
our notions of robustness to the case of approval elections.


We furthermore perform an experimental evaluation of the robustness of
our rules with respect to random swaps. We conclude that, on average, to change the
outcome of an election, one needs to make the most swaps under the
$k$-Borda rule, whereas STV and SNTV (Single Non-Transferable Vote) require
fewest swaps to achieve this result.

The paper is organized as follows. In Section~\ref{sec:prelim} we
provide the necessary background definitions, including the
definitions of the rules that we focus on. In
Section~\ref{sec:robustness_level} we introduce the robustness level
notion and determine robustness level values for our rules.
In Section~\ref{sec:computingRefinements} we link low robustness level
values with the ability to compute refinements of multiwinner rules.
Then, in Sections~\ref{sec:complexity} and~\ref{sec:fpt-algo}, we
introduce the \textsc{Robustness Radius} problem and study its
computational complexity; in the former section we mostly focus on the classic
complexity, whereas in the latter we provide several $\fpt$
algorithms. In Section~\ref{sec:exp-eval} we describe our
experiments. We conclude in Section~\ref{sec:conclusions}.




\begin{table}[t]
\begin{center}
\begin{tabular}{l|c|c}
  \toprule
  Voting Rule              & Robustness Level & Complexity of \textsc{Robustness Radius} \\
  \midrule
  SNTV, Bloc, $k$-Borda ($\p$)             & $1$ & $\p$ \\
  \midrule
  $k$-Copeland ($\p$)      & $1$ & $\np$-hard, FPT($m$), W[1]-hard($n$) \\
  NED ($\np$-hard~\cite{azi-elk-fal-lac-sko:c:multiwinner-condorcet})         & $k$ & $\np$-hard, FPT($m$), W[1]-hard($n$) \\
  \midrule
  STV ($\np$-hard\footnotemark~\cite{con-rog-xia:c:mle})   & $k$ &
                                                                   $\np$-hard(B), FPT($m$), FPT($n$) \\
  $\bordacc$
  ($\np$-hard~\cite{pro-ros-zoh:j:proportional-representation,bou-lu:c:chamberlin-courant,bet-sli-uhl:j:mon-cc})
                           & $k$ & $\np$-hard(B), FPT($m$), FPT($n$) \\
  \bottomrule
\end{tabular}
\end{center}
\caption{\label{tab:results}Summary of our results. For each rule, we provide 
  the complexity of its winner determination.
  The parameters $m$, $n$, and $B$ mean, respectively, the number of candidates, the number of
  voters, and the robustness radius; $\np$-hard($B$) means $\np$-hard even for constant~$B$.}
\end{table}
\footnotetext{For STV there is a polynomial-time algorithm for computing a
single winning committee, but deciding whether a given committee wins is $\np$-hard.}

\section{Preliminaries}\label{sec:prelim}

In this section we describe our model of elections and the voting
rules that we focus on. We assume familiarity with classic and
parameterized computational complexity theory, but we briefly recall the essential
notions from the latter. For each positive integer $m$, we write $[m]$
to denote the set $\{1, \ldots, m\}$.

\paragraph{Elections.}
An election $E = (C,V)$ consists of a set of candidates
$C = \{c_1, \ldots, c_m\}$ and of a collection of voters
$V = (v_1, \ldots, v_n)$. We consider the ordinal election model,
where each voter $v$ is associated with a preference order
$\succ_{v}$, that is, with a ranking of the candidates from the most
to the least desirable one (according to this voter); we sometimes
refer to preference orders as to votes. A multiwinner voting rule
$\calR$ is a function that, given an election $E = (C,V)$ and a
committee size $k$, outputs a set $\calR(E,k)$ of size-$k$ subsets of
$C$, referred to as the winning committees (each of these committees
ties for victory).

\begin{remark}\label{rem:notation}
  Sometimes when we specify a preference order, we write $A \succ B$
  to denote the fact that each candidate in the set $A$ is preferred
  to each candidate in the set $B$, but the particular order of the
  candidates within these sets is irrelevant for the discussion.
\end{remark}

\paragraph{(Committee) Scoring Rules.}
Given a voter $v$ and a candidate $c$, by $\pos_v(c)$ we denote the
position of $c$ in $v$'s preference order (the top-ranked candidate
has position $1$, the following candidate has position $2$, and so
on). A scoring function 
for $m$ candidates is a function
$\gamma_m \colon [m] \rightarrow \reals$ that associates each
candidate-position with a score. 
Examples of scoring functions include (1) the Borda scoring functions,
$\beta_m(i) = m-i$; and (2) the $t$-Approval scoring functions,
$\alpha_t(i)$, defined such that $\alpha_t(i) = 1$ if $i \leq t$ and
$\alpha_t(i)=0$ otherwise ($\alpha_1$ is typically referred to as the
Plurality scoring function). For a scoring function $\gamma_m$, the
$\gamma_m$-score of a candidate $c$ in an $m$-candidate election
$E = (C,V)$ is defined as
$ \score{\gamma_m}_E(c) = \sum_{v \in V}\gamma_m(\pos_v(c)).$

For a given election $E$ and a committee size $k$, the SNTV score of a
size-$k$ committee $S$ is defined as the sum of the Plurality scores
of its members. The SNTV rule outputs the committee(s) with the
highest score (i.e., the rule outputs the committees that consist of
$k$ candidates with the highest plurality scores; 
there may be more than one such committee due to ties).
%
 Bloc and $k$-Borda
rules are defined analogously, but using $k$-Approval and Borda scoring functions,
respectively.
The Chamberlin--Courant rule~\cite{cha-cou:j:cc}
(abbreviated as $\bordacc$, where $\beta$~indicates the Borda scoring function)
also outputs the committees with the highest score, but computes these
scores in a different way: The score of a committee~$S$ in a
vote~$v$ is the Borda score of the highest-ranked member of~$S$, and the
score of a committee in an election is the sum of the scores that it obtains
in all votes.
Given a committee $S$ and a voter $v$, we refer to the
member of $S$ that $v$ ranks highest as his or her representative (in
this committee).

\begin{remark}\label{rem:dissat}
  In a couple of proofs, we use the concept of the
  \emph{dissatisfaction score} that the voters associate with a
  committee according to the $\bordacc$ rule. The dissatisfaction
  score of a voter $v$ for a committee $S$ is equal to $(m-1)$ minus
  the Borda score of the most preferred member of $S$ (according to
  $v$). For example, if~$S$ contains $v$'s top preferred candidate,
  then the dissatisfaction score of $v$ from $S$ is equal to zero; if
  $S$ contains $v$'s second most preferred candidate, then such
  dissatisfaction score is equal to one, and so on. The total
  dissatisfaction score of a committee $S$ is the sum of the
  dissatisfaction scores that the individual voters assign to it.
\end{remark}

SNTV, Bloc, $k$-Borda, and $\bordacc$ are examples of committee
scoring
rules~\cite{elk-fal-sko-sli:j:multiwinner-properties,fal-sko-sli-tal-tal:j:hierarchy-committee,sko-fal-sli:j:axiomatic-committee}.
However, while the first three rules are 
polynomial-time computable, winner determination for $\bordacc$ is~well-known to be
$\np$-hard~\cite{pro-ros-zoh:j:proportional-representation,bou-lu:c:chamberlin-courant}
and $\wtwo$-hard when parameterized by the committee
size~\cite{bet-sli-uhl:j:mon-cc}. Yet, there are many ways of dealing
with this negative result, including $\fpt$-algorithms for other
parameters~\cite{bet-sli-uhl:j:mon-cc}, approximation
algorithms~\cite{bou-lu:c:chamberlin-courant,sko-fal-sli:j:multiwinner},
algorithms for restricted
domains~\cite{bet-sli-uhl:j:mon-cc,sko-yu-fal-elk:j:sc-cc,pet:c:total-unimodularity},
and heuristics~\cite{fal-sli-sta-tal:j:clustering}.

\paragraph{Condorcet-Inspired Rules.}
A candidate $c$ is a Condorcet winner (resp.\ a weak Condorcet winner) if for
each other candidate~$d$, more than (at least) half of the voters prefer $c$ to~$d$.
In the multiwinner case,
a committee is \emph{Gehrlein strongly-stable} (resp.\ \emph{weakly-stable})
if every committee member is preferred to every nonmember by more
than (at least) half of the voters~\cite{geh:j:multiwinner-condorcet},
%
%
and a multiwinner rule is Gehrlein strongly-stable (resp.\ weakly-stable) if
it outputs exactly the Gehrlein strongly-stable (weakly-stable)
committees whenever they exist. For example, let the NED (Number of External
Defeats) score of a committee~$S$ be the number of pairs $(c,d)$ such that (i)
$c$ is a candidate in~$S$, (ii) $d$ is a candidate outside of $S$, and (iii) at
least half of the voters prefer $c$ to $d$. Then, the NED
rule~\cite{coelho:thesis:understanding}, defined to output the
committees with the highest NED score\footnote{Originally, the definition of the
NED rule~\cite{coelho:thesis:understanding} used a ``dual'' definition of the
NED score, and thus it was choosing committees whose NED score was the
smallest.}, is Gehrlein weakly-stable. In contrast, the $k$-Copeland$^0$ rule is
Gehrlein strongly-stable
but not weakly-stable (the Copeland$^\alpha$ score of a candidate $c$,
where $\alpha \in [0,1]$, is the number of 
candidates $d$ such that a majority of the voters prefer $c$ to $d$,
plus $\alpha$ times the number of candidates $e$ such that exactly
half of the voters prefer $c$ to $e$; winning $k$-Copeland$^\alpha$
committees consist of $k$ candidates with the highest scores).
%
%
%
%
Detailed studies of Gehrlein stability mostly focused on the weak
variant of the
notion~\cite{bar-coe:j:non-controversial-k-names,kam:j:stable-rules-again}.
Some recent findings, as well as results from this paper, suggest that
the strong variant is more
appealing~\cite{azi-elk-fal-lac-sko:c:multiwinner-condorcet,sek-sik-xia:c:bundling-condorcet}.
For example, all Gehrlein weakly-stable rules are $\np$-hard to
compute~\cite{azi-elk-fal-lac-sko:c:multiwinner-condorcet}, whereas
there are strongly-stable rules (such as $k$-Copeland$^0$) that are
Polynomial-time computable. (However, we mention that there are
approximation algorithms for some Gehrlein weakly-stable
rules~\cite{sek-sik-xia:c:bundling-condorcet}.)


\paragraph{Single Transferable Vote (STV).}
Let $E = (C,V)$ be an election with $m$~candidates and $n$~voters.  To
select a committee of size $k$, the STV rule proceeds as follows.
First, it computes the quota value $q$; in our case we use the Droop
quota~\cite{Droop1881} $q = \lfloor\frac{n}{k+1}\rfloor+1$. Then it executes up to~$m$ rounds as
follows.
In a single round, it checks whether there is a candidate $c$ who is
ranked first by at least $q$ voters and, if so, then it (i) includes
$c$ into the winning committee, (ii) removes exactly $q$ voters that
rank $c$ first from the election, and (iii) removes $c$ from the
remaining preference orders. If such a candidate does not exist, then
a candidate $d$ that is ranked first by the fewest voters is
removed. Note that this description does not specify \emph{which} $q$
voters to remove or \emph{which} candidate to remove if there is more
than one that is ranked first by the fewest voters. We adopt the
parallel-universes tie-breaking model and we say that a committee wins
under STV if there is any way of breaking such internal ties that
leads to the committee being elected~\cite{con-rog-xia:c:mle}.
 
We can compute \emph{some} STV winning committee by breaking the
internal ties in some arbitrary way, but it is $\np$-hard to decide if
a given committee wins~\cite{con-rog-xia:c:mle}.


\paragraph{Parametrized Complexity.}

%
A parameterized problem is a standard decision problem
where in addition to the problem instance $I$ we also distinguish a
parameter value~$\rho$ (in our problems a typical parameter would be the
number of candidates or the number of voters). An $\fpt$ algorithm for
a parameterized problem is an algorithm that runs in~$f(\rho)|I|^{O(1)}$ time,
where $f$~is some computable function. That is, an
$\fpt$ algorithm can run in exponential time, provided that the
exponential part of the running time depends on the parameter value
only.

The existence of an $\fpt$-algorithm means that, from the parameterized
complexity point of view, the problem is tractable (with respect to a given
parameter). There is also a theory of hardness of parameterized
problems that includes the notion of $\wone$-hardness. If a problem is
$\wone$-hard for a given parameter, then it is widely believed that there is no
$\fpt$-algorithm for the same parameter.
The typical approach to showing that a certain parameterized problem is
$\wone$-hard is to reduce to it a known $\wone$-hard problem, using
the notion of a parameterized reduction. In our case, instead of using
the full power of parameterized reductions, we use standard many-one
reductions that ensure that the value of the parameter in the output
instance is upper-bounded by a function of the parameter of the input
instance.

For more details on parameterized complexity, we point the readers to the
textbooks of Cygan et al.~\cite{CyganFKLMPPS15}, Downey and Fellows~\cite{DF13},
Flum and Grohe~\cite{FG06}, and Niedermeier~\cite{nie:b:invitation-fpt}.



\section{Robustness Levels of Multiwinner
  Rules}\label{sec:robustness_level}

In this section we introduce the notion of the robustness level of a
multiwinner rule and establish its value 
for several prominent rules.  Informally speaking, the robustness level
measures the extent to which a winning committee might change after modifying a
single vote in a given election in the smallest possible way. We formalize
this intuition below (note that our definition takes into account that
a voting rule can output several tied committees).


\begin{definition}\label{def:robustness}
  The \emph{robustness level} of a multiwinner rule $\calR$ for
  committees of size $k$
  is the smallest value $\ell$ such that for each election $E = (C,V)$
  with $|C| \geq k$, each election $E'$ obtained from $E$ by making a
  single swap of adjacent candidates in a single vote, and each
  committee $W \in \calR(E,k)$, there exists a committee
  $W' \in \calR(E',k)$ such that $|W \cap W'| \geq k-\ell$.
\end{definition}

In other words, if we have an $\ell$-robust rule and $W$ is some
winning committee for election~$E$, then after swapping two adjacent
candidates in some vote in $E$ we certainly have a winning committee~$W'$ that
differs from~$W$ in at most $\ell$ candidates (and, indeed,
there are cases where these committees differ in exactly $\ell$
members).\footnote{Consequently, $k$-robustness means that the
committees may be disjoint.} Yet, one may worry what happens if for
the new election we also have some new committees, completely
unrelated to those in~$E$. To deal with this issue, it suffices to
revert the roles of~$E$ and~$E'$ in Definition~\ref{def:robustness}.
For example, if we had $\calR(E,k) = \{W\}$ and
$\calR(E',k) = \{W,W'\}$ where $W$ and $W'$ were disjoint, then
applying Definition~\ref{def:robustness} for $E$ and $E'$ would not
lead to conclusions about the robustness of our rule, but applying it
with the roles of $E$ and $E'$ reversed, and considering committee
$W'$, we would conclude that the rule is $k$-robust.


It turns out that all of the rules that we consider belong to one of
the two extremes: Either they are $1$-robust (i.e., they are very
robust) or they are $k$-robust (i.e., they are possibly very
non-robust).  We start by considering a large class of $1$-robust
rules.

\newcommand{\thmrobust}{Let $\calR{}$ be a voting rule that assigns
  points to candidates and selects those with the highest scores. If a
  single swap in an election affects the scores of at most two
  candidates (possibly decreases the score of one and possibly
  increases the score of the other), then the robustness level
  of~$\calR{}$ is equal to one.}

\begin{proposition}\label{thm:robust}
  \thmrobust
\end{proposition}

\begin{proof}
  Let $E$ be an election, $k$ be a committee size, and $W$ be
  a committee in $\calR(E,k)$. We write~$s(c)$ to denote the individual
  $\calR$-score of a candidate $c$ in $E$. We rename the candidates so
  that (i) $s(c_1) \geq \cdots \geq s(c_m)$ and (ii)
  $W = \{c_1, \ldots, c_k\}$. Now consider an election $E'$ obtained
  from $E$ by a single swap. This swap can increase the score of at
  most one candidate, say~$c_i$, while decreasing the score of at most
  one other candidate, say $c_j$. There are four cases to consider:
  \begin{enumerate}
  \item If $i \leq k$ and $j > k$, then $W$ is still winning in $E'$.
  \item If $i \leq k$ and $j \leq k$, then either $W$ or $\{c_1,
    \ldots, c_{k + 1}\} \setminus \{c_j\}$ is a winning committee in $E'$.
  \item If $i > k$ and $j > k$, then either $W$ or $\{c_1, \ldots,
    c_{k - 1}\} \cup \{c_i\}$ is a winning committee in $E'$.
  \item If $i > k$ and $j \leq k$, then either $W$ or $\{c_1, \ldots,
    c_{k - 1}\} \cup \{c_i\}$ or $\{c_1, \ldots, c_{k + 1}\} \setminus
    \{c_j\}$ or $\{c_1, \ldots, c_{k}\}\setminus~\{c_j\} \cup
    \{c_i\}$ is a winning committee in $E'$.
  \end{enumerate}
  In each case, there is a committee $W' \in \calR(E',k)$ such that
  $|W \cap W'| \geq k-1$ and, so, $\calR$ is $1$-robust.
\end{proof}

Proposition~\ref{thm:robust} suffices to deal with four of our rules:
  SNTV, Bloc, $k$-Borda, and $k$-Copeland$^\alpha$ (for each $\alpha$).
Indeed, it applies to all (weakly) separable committee scoring rules
(i.e., rules defined analogously the our first three rules; see the
work of Elkind et al.~\cite{elk-fal-sko-sli:j:multiwinner-properties}
for a formal definition) and to many multiwinner rules that are
straightforward extensions of single-winner ones (as is the case for
$k$-Copeland$^\alpha$).

\begin{corollary}
  SNTV, Bloc, $k$-Borda, and $k$-Copeland$^\alpha$ (for each $\alpha$)
  are $1$-robust.
\end{corollary}

In contrast, Gehrlein weakly-stable rules are $k$-robust.  This
is quite interesting because for elections with odd numbers of voters,
$k$-Copeland$^\alpha$ rules output Gehrlein weakly-stable committees
whenever they exist~\cite{bar-coe:j:non-controversial-k-names}.  That
is, the non-robustness of Gehrlein weakly-stable rules can be seen as
a consequence of tie-breaking in head-to-head contests between
candidates.

\begin{proposition}\label{pro:gehrlein-level}
  Each Gehrlein weakly-stable rule is $k$-robust, where $k$ is the
  committee size.
\end{proposition}

\begin{proof}
  Consider the following election, described through its majority
  graph (in a majority graph, each candidate is a vertex and there is
  a directed arc from candidate $u$ to candidate $v$ if more than half
  of the voters prefer $u$ to $v$; the classic McGarvey's theorem says
  that each majority graph can be implemented with polynomially-many
  votes~\cite{mcg:j:election-graph}). We form an election with
  candidate set $C = A \cup B \cup \{c\}$, where $A = \{a_1, \ldots,
  a_k\}$ and $B = \{b_1, \ldots, b_k\}$, and with the following
  majority graph: The candidates in $A$ form one cycle, the candidates
  in $B$ form another cycle, and there are no other arcs (i.e., for
  all other pairs of candidates $(x,y)$ the same number of voters
  prefers $x$~to~$y$ as the other way round). We further assume that
  there is a vote, call it $v$, where $c$ is ranked directly below
  $a_1$ (McGarvey's theorem easily accommodates this need).

  In the constructed election, there are two Gehrlein weakly-stable
  committees, $A$ and $B$. To see this, note that if a Gehrlein
  weakly-stable committee contains some $a_i$, then it must also
  contain all other members of $A$ (otherwise there would be a
  candidate outside of the committee that is preferred by a majority
  of the voters to a committee member). An analogous argument holds
  for~$B$.

  If we push $c$ ahead of $a_1$ in vote $v$, then
  a majority of the voters prefers $c$ to $a_1$.
  Thus, $A$ is no longer Gehrlein weakly-stable and~$B$
  becomes the unique winning committee.
  Since (1) $A$ and~$B$ are disjoint,
  (2) $A$ is among the winning committees prior to the swap, and
  (3) $B$ is the unique winning committee after the swap,
  we have that every Gehrlein
  weakly-stable rule is $k$-robust. 
%
%
\end{proof}

We view the above result as particularly negative. The reason is that
Gehrlein weakly stable rules are meant to select groups of
individually excellent candidates, that is, groups of candidates that
perform very well on their own, independently of the other members of
the winning committee.  Such rules are useful, for example, in sport
competitions or various other contests to select
finalists~\cite{bar-coe:j:non-controversial-k-names,elk-fal-sko-sli:j:multiwinner-properties}
(for a more detailed discussion of individual excellence, diversity,
and proportionality, we point to the overview of
Faliszewski et al.~\cite{fal-sko-sli-tal:b:multiwinner-voting}). Thus, a single
swap of two adjacent candidates in a single preference order certainly should
not result in a rule declaring \emph{all} candidates
that were previously seen as ``individually best'' to no longer be
``good enough.'' On the other hand, we view the following
results---where we show that $\beta$-CC and STV are only
$k$-robust---as less negative. Indeed, $\beta$-CC aims at choosing a
diverse committee that covers the views of as many voters as possible,
whereas STV seeks a committee that represents these views
proportionally. While the fact that a single swap can replace the
whole committee seems undesirable, it is natural that candidates'
memberships in diverse/proportional committees are correlated, so
replacing one of them can lead to a cascading effect of replacing them
all.  Further, it is quite plausible that there are several disjoint
committees that achieve diversity or proportionality to nearly the
same extent (see, e.g., the experiments of Elkind et
al.~\cite{elk-fal-las-sko-sli-tal:c:2d-multiwinner}).

\begin{example}\label{ex:correlated} {To illustrate the issue of correlation between the members of
    a diverse/proportional committee, consider the following example. We have $104$
    candidates, $a$, $b$, $c$, $d$, $e_1, \ldots, e_{100}$, and four
    voters $v_1$, $v_2$, $v_3$, $v_4$ with the following preference orders:
    \begin{align*}
      &v_1 \colon a \pref c \pref e_1 \pref \cdots \pref e_{100} \pref b \pref d,&
      &v_2 \colon b \pref d \pref e_1 \pref \cdots \pref e_{100} \pref a \pref c,\\
      &v_3 \colon a \pref d \pref e_1 \pref \cdots \pref e_{100} \pref b \pref c,& 
      &v_4 \colon b \pref c \pref e_1 \pref \cdots \pref e_{100} \pref a \pref d.&
    \end{align*}
    It is natural to select committee $\{a,b\}$ as a diverse committee
    of size two because then each voter has his or her most desirable
    representative in the committee (indeed, this committee can also
    be seen as proportional). Yet, if for some reason we had to remove
    $b$ from the committee, then it might also make sense to remove
    $a$ from it and choose committee $\{c,d\}$ instead. This way each
    voter would still have a nearly perfect representative. On the
    contrary, choosing one of the committees $\{a,c\}$ or $\{a,d\}$
    would mean that one voter would rank both members of the committee
    at the two bottom positions (including the candidates
    $e_1, \ldots, e_{100}$ would also lead to a committee that is less
    desirable than $\{c,d\}$).}
\end{example}

The next two propositions build on ideas similar to those used in the proof
of Proposition~\ref{pro:gehrlein-level}, but they are targeted for their
respective rules.

\begin{proposition}\label{prop:beta-CC_chaotic}
  $\bordacc$ is $k$-robust, where $k$ is the committee size.
\end{proposition}

\begin{proof}
  We form an election with candidate set~$C:=A \cup B \cup \{x, y\}$,
  where $|A|=k-1$, $|B|=k-1$, and with $2k-1$ voters. The first voter
  has preference order
 \[
   v_1:\ x \succ y \succ A \succ B,
 \]
 while the remaining pairs of voters, one for each $i \in [k - 1]$, have
 preference orders
 \begin{align*}
   v_{2i}:\ & a_i \succ x \succ A\setminus\{a_i\} \succ B \succ y,\\
   v_{2i+1}:\ & b_i \succ y \succ A \succ B\setminus\{b_i\} \succ x.
 \end{align*}
 
 Observe that the only winning committee is $\{x\} \cup B$. To see
 this, note that $\{x\} \cup B$ has dissatisfaction score  of only~$k-1$
 (recall Remark~\ref{rem:dissat}).  Further, each voter has a
 different favorite candidate, there are $2k-1$~voters, and the
 committee size is only~$k$. Hence, $k-1$ is the lowest possible
 dissatisfaction score value. Further, each committee with dissatisfaction score
 $k-1$ must contain~$k$ of the ``favorite'' candidates from
 $\{x\} \cup A \cup B$ and every voter that is not represented by her
 favorite candidate must be represented by her second choice. Now,
 if~$x$ were not in the committee, voter~$v_1$ could not be
 represented by its second choice because
 $y \notin \{x\} \cup A \cup B$. So, $x$~belongs to each winning
 committee and $y$ does not belong to any of them.  As a consequence,
 all remaining members of the winning committee are from~$B$ since
 only voters~$v_{2i}$, $i \in [k-1]$, can be represented by their second
 choices.

 If we swap~$x$ and~$y$ in the first vote, then, following analogous
 argumentation, the unique winning committee becomes $\{y\} \cup A$.
 Finally, we mention that the construction above works for every
 committee size.
\end{proof}

\begin{proposition}\label{prop:stv-chaotic}
  STV is $k$-robust, where $k$ is the committee size.
\end{proposition}
\begin{proof}
  Let us fix the committee size $k$ and consider a set of $m = 2k$
  candidates $C:=A \cup B$, where $A = \{a_1, \ldots, a_k\}$ and
  $B = \{b_1, \ldots, b_k\}$. For each pair of candidates $a_i \in A$
  and $b_j \in B$, we form $k+1$ voters with preference order
  $a_i \pref b_j \pref \cdots$ and $k+1$ voters with preference order
  $b_j \pref a_i \pref \cdots$. Let $v$ be one of the voters with
  preference order $b_1 \pref a_1 \pref \cdots$. We modify $v$'s vote
  by swapping $b_1$ and $a_1$ and we refer to $v$ as the pivotal vote.
  Altogether, we have $n = 2k^2(k+1)$ voters, so the STV quota value
  is $q = \lfloor \frac{2k^2(k+1)}{k+1} \rfloor + 1 = 2k^2+1$.
  Initially, each candidate in $A \cup B$, except for $a_1$ and $b_1$,
  has the same Plurality score equal to $k(k+1)$, candidate $a_1$ has
  one point more, and candidate $b_1$ has one point less.

  We claim that STV chooses $A$ as the unique winning committee.
  Indeed, all the candidates have fewer Plurality points than the
  quota value, so in the first round STV removes candidate $b_1$,
  whose score is lowest. As a consequence, the scores of all
  candidates from~$A$ become~$k(k+1) + (k+1)$, whereas the scores of
  all candidates from~$B$ do not change. In the following rounds there are
  two possibilities: Either no candidate meets the quota and some
  member of $B$ is removed (in effect, the scores of all the
  candidates in $A$ increase by the same amount while the scores of
  candidates in $B$ do not change) or all the candidates in $A$ meet
  the quota (between the first round and the current one, all members
  of $A$ always have the same Plurality score and the scores of the
  candidates from~$B$ never increase). If the latter happens, then in
  the following rounds all members of $A$ are selected for the
  committee. During these rounds the scores of candidates from $B$
  increase (as members of~$A$ are removed from the election and
  included in the committee), but no member of $B$ ever obtains score
  higher than $n - qk = 2k^2(k+1) - (2k^3+k) = 2k^2-k$, which is lower
  than the quota value.

  Now, let us consider what happens when we swap $b_1$ and $a_1$ in
  the pivotal vote. As a consequence, all the candidates have the same
  score and STV eliminates some arbitrary candidate in the first
  round. If it eliminates some member of $B$, then---by the same
  reasoning as above---it chooses committee $A$. However, if it
  eliminates a member of $A$, then, by the same token, it chooses
  committee $B$. As (a) $A$ and $B$ are disjoint, (b) only $A$ is
  winning before the swap, and (c) both $A$ and $B$ are winning after
  the swap, we conclude that STV is $k$-robust.
\end{proof}

\newcommand{\propbccstv}{Both $\bordacc$ and STV are $k$-robust.}
%
%

%

So far we have only seen voting rules that are either $1$-robust or
$k$-robust. Indeed, we are not aware of any classical rule with
robustness level between these two extremes, but we conclude this section
by showing that there are hybrid multi-stage rules with arbitrary
robustness levels.
For example, the rule which first elects half of the committee as
$k$-Borda does and then the other half as $\beta$-CC does has
robustness level of roughly $k/2$ (such a rule is not completely
artificial---for example, Kocot et
al.~\cite{koc-kol-elk-fal-tal:c:multigoal} use a similar strategy for
finding committees that perform well according to both $k$-Borda and
the Chamberlin--Courant rule).


\newcommand{\propalllevels}{For each committee size $k$ and each
  $\ell \in [k]$ there is a multiwinner rule that is $\ell$-robust for
  committees of size $k$.}

\begin{proposition}\label{prop:all-levels}
  \propalllevels
\end{proposition}

\begin{proof}
  Since we know that, for example, $k$-Borda and $\beta$-CC are, respectively,
  $1$-robust and $k$-robust for all possible committee sizes, it
  suffices to show rules with robustness levels between $2$ and
  $k-1$. We fix a committee size $k > 1$ and let $\ell$ be an integer
  between $2$ and $k-1$. We let $\ell' := \ell -1$ and we define a
  voting rule that first selects $k-\ell'$~committee members exactly
  as $(k-\ell')$-Borda would, and then selects further~$\ell'$
  candidates that, jointly, maximize the $\beta$-CC score of the whole
  committee. We refer to this rule as $(k-\ell')$-Borda/$\ell'$-CC. We
  will show that this rule is $\ell$-robust (however, it will be
  easier to express this as $(\ell'+1)$-robustness).

  \newcommand{\borda}{\mathrm{B}}
  \newcommand{\cc}{\mathrm{CC}}

  We first show that our rule is at least $(\ell'+1)$-robust. Let $E$
  be some election (with at least $k$ candidates), let $W$ be a
  winning committee for this election, and let $W_\borda$ be its part
  that is selected using $(k-\ell')$-Borda. Let $E'$ be an election
  obtained from $E$ by swapping two adjacent candidates. Since
  $(k-\ell')$-Borda is $1$-robust, our rule certainly has some winning
  committee $W'$ for $E'$ whose $(k-\ell')$-Borda part differs from
  $W_\borda$ in at most one candidate. As a consequence, for this
  committee it must be the case that $|W'\cap W|\ge k-\ell'-1$.
  This shows that our rule is at least $(\ell'+1)$-robust.
  
  Next we show that indeed there are elections where a single swap
  leads to replacing $\ell'+1$ candidates. To this end, we use an
  election very similar to that used in
  Proposition~\ref{prop:beta-CC_chaotic}.  We let the candidate set
  be~$C:=A \cup B \cup \{x, y\} \cup D$, where $|A|=|B|=\ell'$ and
  $|D|=k-\ell'-1$.  We form the voters as follows:
  \begin{enumerate}
  \item We construct voter $v_1$ with preference order
    $v_1:\ x \succ y \succ A \succ B \succ D$.
  \item For each $i \in [\ell]$, we construct a pair of voters with
    preference orders:
    \begin{align*}
      v_{2i}:\ & a_i \succ x \succ y \succ A\setminus\{a_i\} \succ B \succ D,\\
      v_{2i+1}:\ & b_i \succ y \succ x \succ A \succ B\setminus\{b_i\} \succ D.
    \end{align*}
  \item For each $j \in [k-\ell-1]$ we construct sufficiently many
    pairs of voters with preference orders:
    \begin{align*}
      w_j:\ & d_j \succ D\setminus\{d_j\} \succ x \succ y \succ B \succ A,\\
      w'_j:\ & d_j \succ D\setminus\{d_j\} \succ y \succ x \succ A \succ B,
    \end{align*}
    so that all candidates in~$D$ have higher Borda scores than
    all other ones, and candidates~$x$ and~$y$ have higher Borda
    scores than all members of~$A$ and~$B$.
  \end{enumerate}
  As a consequence, in the~$(k-\ell')$-Borda phase our rule selects
  all candidates from~$D$ and candidate~$x$ ($x$ has higher score
  than~$y$ due to voter~$v_1$). Then, in the~$\ell'$-CC phase, our
  rule selects all~$\ell'$~candidates from~$B$. To see this, we
  first note that all the voters from the third group already have
  their top-ranked candidates in the committee, and, so, do not affect
  the selection of the remaining candidates; then we reuse the
  reasoning from Proposition~\ref{prop:beta-CC_chaotic}.
  
  If we swap candidates $x$ and $y$ in vote $v_1$, then candidate~$y$
  will be selected instead of candidate~$x$
  in the $(k-\ell')$-Borda phase, and all the candidates from~$A$ will
  be selected as the remaining~$\ell'$~committee members in the
  $\ell'$-CC phase.
  All in all, prior to swapping $x$ and $y$ in vote $v_1$, our
  election has a unique winning committee $D \cup \{x\} \cup B$, but
  after the swap $D \cup \{y\} \cup A$ becomes the unique winning
  committee. These committees differ in exactly $\ell' + 1 = \ell$
  candidates, which completes the proof.
\end{proof}

\section{Computing Refinements of Robust Rules}\label{sec:computingRefinements}

It turns out that the dichotomy between $1$-robust and $k$-robust
rules is strongly connected to the one between polynomial-time computable rules and
those that are $\np$-hard. To make this claim formal, we need the
following definition.

\begin{definition}\label{def:s-e}
  A multiwinner rule $\calR$ is \emph{scoring-efficient} if the
  following holds:
\begin{enumerate}

\item There is an algorithm that given three positive integers $n$,
  $m$, and $k$ ($k \leq m$) outputs (i)~an election $E$ with $n$ voters
  and $m$ candidates, and (ii)~a size-$k$ committee $S$, such that
  $S \in \calR(E,k)$. This algorithm runs in polynomial time with
  respect to $n$, $m$, and $k$.

\item There is a polynomial-time computable function $f_\calR$ that
  for each election $E$, committee size~$k$, and committee $S$,
  outputs score $f_\calR(E,k,S)$ of committee $S$ in election $E$, so that $\calR(E,k)$
  consists exactly of the committees with the highest $f_\calR$-score.
\end{enumerate}
\end{definition}
The first condition from Definition~\ref{def:s-e} is quite
straightforward to satisfy. For example, for most natural voting rules
it is easy to compute a winning committee for an election where all
voters rank the candidates identically. In particular, this holds
for \emph{weakly unanimous} rules.

\begin{definition}[Elkind et al.~\cite{elk-fal-sko-sli:j:multiwinner-properties}]
  A rule $\calR$ is \emph{weakly unanimous} if for each election
  $E = (C, V)$ and each committee size $k$, if each voter ranks the
  same set~$W$ of $k$~candidates on top (possibly in a different
  order), then $W \in \calR(E, k)$.
\end{definition}

All voting rules which we consider in this paper are weakly
unanimous 
(indeed,
voting rules which are not weakly unanimous are somewhat
``suspicious''). 
%
%
Further, all our rules, except STV, satisfy the second condition from
Definition~\ref{def:s-e}. For example, while winner determination for
$\beta$-CC is indeed NP-hard, computing the score of a given committee
can be done in polynomial time. With this background, we are ready to
state and prove the main result of this section.

\begin{theorem}\label{theorem:robust implies poly wd}
  Let $\calR$ be a $1$-robust scoring-efficient multiwinner rule.
  Then there is a rule~$\calR'$
  such that for each election~$E$ and committee size~$k$
  we have $\calR'(E,k) \subseteq \calR(E,k)$
  and the winner determination for~$\calR'$ is polynomial-time computable.
%
\end{theorem}

\begin{proof}
  Our proof proceeds by showing a polynomial-time algorithm that given
  an election $E$ and committee size $k$ finds a single committee $W$
  such that $W \in \calR(E,k)$; we define $\calR'(E,k)$ to
  output~$\{W\}$.

  Let~$E = (C,V)$ be our input election and let $k$ be the size of the
  desired committee. Let $E' = (C,V')$ be an election with $|V'| =
  |V|$, whose existence is guaranteed by the first condition of
  Definition~\ref{def:s-e}, and let $S'$ be a size-$k$ $\calR$-winning
  committee for this election, also guaranteed by
  Definition~\ref{def:s-e}.
%
  The idea is to transform $E'$ into $E$ by a sequence of swaps, while
  at the same time transforming committee $S'$ to an $\calR$-winning
  committee for $E$ (for ease of presentation, we assume that all
  elections in our discussion contain the same voters, but with
  possibly different preference orders).

  Let $E_0, E_1 \ldots, E_t$ be a sequence of elections such that
  $E_0 = E'$, $E_t = E$, and for each integer $i \in [t]$, we obtain
  $E_i$ from~$E_{i-1}$ by (i) finding a voter $v$ and two candidates
  $c$ and $d$ such that in $E_{i-1}$ voter $v$ ranks $c$ right ahead
  of $d$, but in $E$ voter $v$ ranks $d$ ahead of $c$ (although not
  necessarily right ahead of $c$), and (ii) swapping $c$ and $d$ in
  $v$'s preference order. We note that at most $|C||V|^2$ swaps
  suffice to transform $E'$ into $E$ (i.e., $t \leq |C||V|^2$).

  For each $i \in \{0,1, \ldots, t\}$, we find a committee $S_i\in
  \calR(E_i,k)$. We start with $S_0 = S'$ (which satisfies our
  condition) and for each $i \in [t]$, we obtain $S_i$ from~$S_{i-1}$
  as follows:
  Since $\calR$ is $1$-robust, we know that at least one committee~$S''$
  from the set $\{S'' \mid |S_{i-1} \cap S''| \geq
  k-1\}$ is winning in $E_i$.
  We try each committee $S''$ from this set and compute its
  $f_\calR$-score (recall Condition 2 of Definition~\ref{def:s-e}).
  The committee with the highest $f_\calR$-score must be winning in
  $E_i$ and we set $S_i$ to be this committee (by
  Definition~\ref{def:s-e}, computing the $f_\calR$-scores is a
  polynomial-time task).

  Finally, we output $S_t$. By our arguments, we have that $S_t \in
  \calR(E,k)$. 
  Clearly, our procedure runs in polynomial time.
%
\end{proof}

Theorem~\ref{theorem:robust implies poly wd} generalizes to the case
of $r$-robust rules for constant $r$; our algorithm simply has to try
more (but still polynomially many) committees $S''$.

\begin{corollary}
  Let $r$ be a fixed positive integer and let $\calR$ be an $r$-robust
  scoring-efficient multiwinner rule.  Then there is a polynomial-time
  computable rule $\calR'$ such that for each election $E$ and
  committee size $k$ we have $\calR'(E,k) \subseteq \calR(E,k)$.
\end{corollary}

Note how Theorem~\ref{theorem:robust implies poly wd} relates
to single-winner rules, which can be seen as multiwinner rules for
$k=1$. All such rules are $1$-robust, but for those with $\np$-hard
winner determination problems, even computing the candidates' scores is
$\np$-hard (see, e.g., the survey of Caragiannis et
al.~\cite{car-hem-hem:b:dodgson-young}), so
Theorem~\ref{theorem:robust implies poly wd} does not apply.
Indeed, the fact that committee scores are polynomial-time computable
for many typical $\np$-hard multiwinner rules is a significant
difference between them and $\np$-hard single-winner rules.

%

\section{Complexity of Computing the Robustness Radius}\label{sec:complexity}

In the \textsc{Robustness Radius} problem we are given an election and
we ask whether it is possible to change its result by performing a
given number of swaps of adjacent candidates. Intuitively, the more
swaps are necessary, the more robust a particular election is.

\begin{definition}
  Let $\calR$ be a multiwinner rule. In the $\calR$ \textsc{Robustness
    Radius} problem we are given an election $E = (C,V)$, a committee
  size~$k$, and an integer~$B$. We ask if it is possible to obtain an
  election $E'$ by making at most $B$ swaps of adjacent candidates
  to the votes in $E$ so that
  $\calR(E',k) \neq \calR(E,k)$.
\end{definition}

The \textsc{Robustness Radius} problem is strongly connected to some other
problems studied in the literature.
Specifically, in the \textsc{Destructive Swap Bribery} problem (\textsc{DSB} for
short) we ask if it is possible to preclude a particular candidate from winning
by making a given number of
swaps~\cite{elk-fal-sli:c:swap-bribery,shi-yu-elk:c:robustness,kac-fal:c:destructive-shift-bribery}.
\textsc{DSB} was already used to study robustness of single-winner
election rules by Shiryaev et al.~\cite{shi-yu-elk:c:robustness}.
We decided to give our problem a different name, and not to refer to
it as a multiwinner variant of \textsc{DSB}, because we feel that in
the latter the goal should be to preclude a given candidate from being
a member of any of the winning committees, instead of changing the
outcome in any arbitrary way. In this sense, our problem is very
similar to the \textsc{Margin of Victory}
problem~\cite{mag-riv-she-wag:c:stv-bribery,car:c:margin-of-victory,xia:margin-of-victory,blom2017towards},
which is also related to the notions of approximation for sublinear
winner determination algorithms and sampling of
elections~\cite{dey-nar:c:sampling-margin-of-victory,fil-tal:c:sampling-monitoring}; the \textsc{Margin of
  Victory} problem has the same goal, but instead of counting single
swaps, it counts how many votes are changed.

We find that \textsc{Robustness Radius} tends to be computationally
challenging. Indeed, we find polynomial-time algorithms only for the
simplest of our rules, SNTV, Bloc, and $k$-Borda.

\begin{theorem}\label{pro:weak-sep}
  \textsc{Robustness Radius} is solvable in polynomial time for SNTV, Bloc,
  and $k$-Borda.
\end{theorem}
\begin{proof}
  Each of our rules proceeds by computing an individual score for each
  of the candidates (based on this candidate's positions in the
  preference orders of the voters) and by letting the winning
  committees consist of the candidates with the highest scores. We
  first describe a general strategy for dealing with rules of this
  form and then show how to implement this strategy for SNTV, Bloc,
  and $k$-Borda.

  Let $\calR$ be one of our rules, let $E = (C,V)$ be an election with
  $C = \{c_1, \ldots, c_m\}$ and $V = (v_1, \ldots, v_n)$, and let $k$ be
  the committee size. Let $s(c_1), \ldots, s(c_m)$ be the individual scores
  of the candidates $c_1, \ldots, c_m$. Without loss of generality, assume that $s(c_1) \geq \cdots \geq s(c_m)$.  We are interested in
  computing a shortest sequence of swaps of adjacent candidates that
  transforms election~$E$ into some election~$E'$ such that
  $\calR(E,k) \neq \calR(E',k)$. We consider two cases:
  \begin{enumerate}
  \item There is a unique winning committee in election $E$.
  \item There are several tied winning committees in election $E$.
  \end{enumerate}
  We focus on the case with a unique winning committee first.  The
  winning committee is $W = \{c_1, \ldots, c_k\}$ and we have that
  $s(c_k) > s(c_{k+1})$. Consider some arbitrary sequence of swaps
  that transforms $E$ into some election $E'$ such that
  $\calR(E,k) \neq \calR(E',k)$, and consider the first swap after
  performing which the set of winning committees changes. Prior to
  executing this swap, each of the candidates $c_1, \ldots, c_k$ had
  his or her individual score higher than each of the candidates
  $c_{k+1}, \ldots, c_m$, whereas afterward some candidate from the
  latter group had his or her individual score at least as high as one
  of the members of the former group. Thus to find the shortest
  sequence of swaps that changes the result of election $E$, it
  suffices to find the shortest sequence of swaps that ensures that
  some candidate from the set $C \setminus W$ has at least as high
  score as some candidate from committee~$W$.

  \newcommand{\abo}{{\mathrm{above}}}
  \newcommand{\bel}{{\mathrm{below}}}
  \newcommand{\eql}{{\mathrm{equal}}}
  
  Now let us consider the case where there are several winning
  committees. It must be the case that $s(c_k) = s(c_{k+1})$ and we can
  partition the set of candidates into three sets, depending on the
  relation of their score to that of $c_k$:
  \begin{align*}
    C_\abo = \{ c_i \mid s(c_i) > s(c_k) \},
    &&C_\eql = \{ c_i \mid s(c_i) = s(c_k) \},
    &&C_\bel = \{ c_i \mid s(c_i) < s(c_k) \}.
  \end{align*}
  Each $\calR$-winning committee for election $E$ consists of all the
  candidates from the set $C_\abo$ and an arbitrary subset of
  $k-|C_\abo|$ candidates from $C_\eql$.  As in the previous case, let
  us consider a sequence of swaps that transforms election $E$ into
  one with a different set of winning committees, and consider the
  first swap after which the set of winning committees changes. The
  effect of this swap must be that one of the following situations
  happens:
  \begin{enumerate}
  \item Not all candidates in $C_\eql$ have the same score.
  \item All candidates in $C_\eql$ have the same score, but some
    candidate in $C_\abo$ obtains score at most the one of
    the candidates in~$C_\eql$.
  \item All candidates in $C_\eql$ have the same score, but some
    candidate in $C_\bel$ obtains score at least the one of the
    candidates in~$C_\eql$.
  \end{enumerate}
  So, to be able to find the shortest sequence of swaps that changes
  the result of election $E$, it suffices to be able to find the
  shortest sequence of swaps that ensures that one given candidate has
  score higher (or equal) than some other given candidate. For
  example, to deal with the possibility that the shortest sequence of
  swaps that changes the election result leads to some members of
  $C_\eql$ having different scores, it suffices to try each pair
  $p, d$ of distinct candidates from $C_\eql$ and find the shortest
  sequence of swaps that ensures that the score of~$p$ is greater than
  that of $d$. We consider other possible scenarios listed above
  analogously.

  As a consequence of the above reasoning (for both the case of a
  unique winning committee and the case of several winning
  committees), to prove our theorem it remains to show for each of our
  three rules a polynomial-time procedure that given two candidates,
  $p$ and~$d$, finds the shortest sequence of swaps that ensures that
  the score of~$p$ is greater than (or, at least) the score of~$d$. We provide
  such procedures below (we focus on the case of ensuring that $p$'s score is
  at least that of~$d$; adapting our reasoning to the case of
  ensuring that $p$~has strictly greater score than~$d$ is straightforward):

  \begin{description}
  \item[SNTV.] For the case of SNTV, our procedure works as
    follows. We guess three nonnegative numbers, $B_1$, $B_2$, and
    $B_3$. We find $B_1$ votes where $d$ is ranked first and $p$ is
    ranked as high as possible, and we shift $p$ to the top position
    (so $d$ loses his or her Plurality point and $p$~gains it). Then
    we find $B_2$ votes where $p$ is ranked as high as possible (but
    not on the first position), and we shift $p$ to the top
    position. Finally, we find $B_3$ votes where $d$ is ranked first,
    and we shift him or her down by one position in each of these
    votes. (If at any point of this algorithm we do not find
    sufficiently many voters with a given property, we drop this guess
    of $B_1$, $B_2$, and $B_3$.)  We check if as a consequence of our
    swaps $p$'s score is at least the same as that of~$d$ and, if so, we
    record the number of swaps performed. Finally, after considering
    all possible $O(n^3)$ guesses of $B_1$, $B_2$, and $B_3$, we
    output the lowest number of swaps recorded (note that for at least
    one guess our procedure must have succeeded; e.g., when it ensured
    that all voters rank $p$ on top).

  \item[Bloc.] We proceed in the same way as in the case of SNTV, but
    our guesses are a bit more involved.  First, we partition the
    voters into four groups:
    \begin{enumerate}
    \item Voters who neither give a point to $p$ nor to $d$.
    \item Voters who give a point to $p$ but not to $d$.
    \item Voters who give a point to $d$ but not to $p$.
    \item Voters who give points to both $p$ and $d$.
    \end{enumerate}
    We guess numbers $B_1$, $B'_3$, $B''_3$, and $B_4$ of voters, 
    whose preference orders we will modify (note that there
    is no point in affecting the voters in the second group, but there
    are two ways of modifying the preference orders of the voters in the
    third group). For the first group, we execute the smallest number
    of swaps that ensures that $B_1$ voters give a point to~$p$.  For
    the third group, we execute the smallest number of swaps that
    ensures that $B'_3$ voters give a point to $p$ and that $B''_3$
    voters do not give a point to $d$ (note that these operations are,
    in essence, independent). For the fourth group, we execute the
    smallest number of swaps that ensures that $B_4$ voters do not
    give a point to $d$.

  \item[$\boldsymbol k$-Borda.] We perform the following operation
    until the score of $p$ is at least the same as that of~$d$: We
    find a vote where $p$ is ranked below $d$, but the difference
    between their positions is smallest, and we shift $p$ one position
    higher (possibly passing $d$, if in this vote $p$ is ranked just
    below $d$).  Note that if the score of $p$ is lower than that of
    $d$, then there must be a vote where $p$ is ranked below $d$, each
    swap decreases the difference between the scores of~$p$ and~$d$ by
    one~point or by two~points (if $p$ passes $d$), and our strategy of choosing
    swaps ensures the highest number of swaps of value~two.
  \end{description}
  This completes the proof.
\end{proof}

\newcommand{\thmeasyrr}{\textsc{Robustness Radius} is computable in
  polynomial time for SNTV, Bloc, and $k$-Borda.}


The rules in Theorem~\ref{pro:weak-sep} are all $1$-robust, but
not all $1$-robust rules have efficient \textsc{Robustness Radius}
algorithms. In particular, a simple modification of a proof of
Kaczmarczyk and
Faliszewski~\cite[Theorem~6]{kac-fal:c:destructive-shift-bribery}
shows that for $k$-Copeland$^\alpha$ rules (which are $1$-robust) we
obtain $\np$-hardness. We also obtain a general $\np$-hardness result
for all Gehrlein weakly-stable rules.

\begin{corollary}
  $k$-Copeland \textsc{Robustness Radius} is $\np$-hard.
\end{corollary}

\newcommand{\thmgwshard}{\textsc{Robustness Radius} is $\np$-hard for
  each Gehrlein weakly-stable rule.}

\begin{theorem}\label{thm:gws-hard}
  \thmgwshard
\end{theorem}
\begin{proof}
  We reduce from the $\np$-hard \textsc{Exact 3-Set Cover}
  problem~\cite{gar-joh:b:int} where we are given a
  set~$X=\{x_1, \ldots, x_{3h}\}$ of elements and a
  set~$\mathcal{S} = \{S_1, \ldots, S_m\}$ of triplets of elements
  of~$X$. We ask for $h$~triplets that, together, contain all 
  elements of~$X$. In the following reduction we assume that every
  element occurs in exactly three triplets; this variant of the
  problem remains $\np$-hard \cite{gon85}.

  Our reduction proceeds as follows.  For each element~$x \in X$, we
  have an \emph{element candidate}~$c(x)$ (for a given set $X'$ of the
  elements, $X' \subseteq X$, we write $c(X')$ to denote the set of
  element candidates that correspond to the members of $X'$; in
  particular, $C(X)$ means the set of all element candidates). We
  will have $2m+8h$ voters and for each of them we introduce $4h+1$
  distinct dummy candidates. We write $\mathcal{D}$ to denote the set
  of all these dummy candidates, and for each voter~$v$ we write
  $D(v)$ to denote the set of dummy candidates associated with
  $v$. Further, we also have two special candidates, $p$ and
  $d$. Altogether, we have $2 + 3h + (4h+1)(2m + 8h)$~distinct
  candidates, collected in the set:
  \[
    C = \{p,d\} \cup c(X) \cup \mathcal{D}.
  \]
  Using the notation introduced in~Remark~\ref{rem:notation}, we form
  the following $2m+8h$ voters (in each preference order, ellipses
  represent all the candidates not mentioned explicitly, ordered in
  an arbitrary way):
  \begin{enumerate}
  \item For each triplet~$S \in \mathcal{S}$, there are two voters:
    \begin{align*}
      v_S \colon \quad& d \succ c(S) \succ p \succ D(v_S) \succ \ldots,\\
      \bar{v}_S \colon \quad& c(X) \setminus c(S) \succ D(\bar{v}_S) \succ p \succ
                              d \succ \ldots. 
    \end{align*}
  \item For each $i \in [h-1]$, there is a voter with the following
    preference order:
    $$v_i \colon \quad d \succ D(v_i) \succ p \succ c(X) \succ \ldots.$$
  \item For each $i \in [h+1]$, there is a voter with the following
    preference order:
    $$v'_i \colon \quad d \succ c(X) \succ D(v'_i) \succ p \succ \ldots.$$
  \item For each $i \in [3h]$, there are two special voters:
    \begin{align*}
      v^{*}_i \colon \quad& d \succ c(X) \succ D(v^*_i) \succ p \succ \ldots,\\
      \bar{v}^*_i \colon \quad& p \succ d \succ D(\bar{v}^*_i) \succ c(X) \succ
                                \dots.
    \end{align*}
  \end{enumerate}

  We form an instance of the~\textsc{Robustness Radius} that contains
  an election with the candidates and voters described above,
  committee size~$k=1$, and the number of swaps set to $B=4h$.

    \begin{figure}
    \centering%
    \scalebox{.7}{%
      \begin{tikzpicture}[->,shorten >=1pt,auto,node distance=4cm,
        semithick]
    \tikzstyle{cand}=[draw, circle, minimum width=1.3cm]
    \node[cand] (d) {$d$};
    \node[cand] (p) [right of=d] {$p$};
    \node[cand] (X) [right of=p] {$c(X)$};
    \node[cand] (D) [below of=p] {$\mathcal{D}$};
  
    \path (d) edge  node {$2h$} (p)
              edge[bend left=20]  node {$6+8h$} (X)
	      edge node[sloped, xshift=-1.3cm] {$\geq 2m+8h-2$} (D)
	  (p) edge node [sloped, xshift=-1.45cm, yshift=0.3cm] {$\geq 2m+8h-2$} (D)
	  (X) edge node {$2$} (p)
	      edge node [sloped, xshift=-1.5cm] {$\geq 2m+8h-2$} (D);
   \end{tikzpicture}
   }
   \caption{\label{fig:gehrlein-weakly-stable-reduction}A (simplified)
     majority graph of the election constructed by the reduction in
     the proof of~Theorem~\ref{thm:gws-hard}. All dummy
     candidates~$\mathcal{D}$ and all element candidates~$c(X)$ are
     contracted to a single vertex. All arcs within the contracted
     vertices are neglected.}
   \end{figure}

   We present the constructed election visually as a (slightly
   simplified) weighted majority graph in
   Figure~\ref{fig:gehrlein-weakly-stable-reduction}. In this graph,
   each vertex corresponds either to a single candidate or to a set of
   candidates. If we have an edge from a vertex associated with
   candidate $c$ to a vertex associated with candidate $c'$, with
   weight $w$, then it means that $w$ more voters prefer $c$ to $c'$
   than the other way round.
   For example, there is an arc with weight~$6+8h$ pointing from
   candidate~$d$ to a vertex associated with~$c(X)$.  This arc
   indicates that for every element candidate~$c(x)$, the set of
   voters that prefer $d$ to $c(x)$ contains $6+8h$ more voters than
   the set of voters who prefer $c(x)$ to $d$. To see that this indeed
   is the case, note that every voter in groups~$2$, $3$, and $4$
   prefers~$d$ to~$c(x)$; hence we have~$8h$ voters who prefer~$d$
   to~$c(x)$. In group~$1$ (of~$2m$ voters), $d$~is preferred to~$x$
   by exactly $m+3$~voters. Thus, in this group, six more voters
   prefer~$c(x)$ to $d$ than the other way round.  The computation is
   analogous for all other element candidates and, thus,
   candidate~$d$'s winning margin over each of them is~$6 + 8h$.

   Let us now show that the reduction is correct.  Note that as the
   committee size is one, if some candidate is a Condorcet winner, then
   every Gehrlein weakly-stable rule outputs a single winning
   committee, containing exactly this candidate. Similarly, if there
   are weak Condorcet winners in the election, then the winning
   committees are exactly those singletons that contain them. In our
   election, $d$ is a Condorcet winner (indeed, in
   Figure~\ref{fig:gehrlein-weakly-stable-reduction} there are arcs
   from~$d$ to every other vertex) and, so, committee~$\{d\}$ wins
   uniquely.

   Let us assume that there is an exact cover of~$X$ with $h$~triplets
   from~$\mathcal{S}$, and let~$I = \{i_1, \ldots, i_h\}$ be the set
   of indices of these triplets (formally, we have that
   $\bigcup_{i \in I}S_i = X$).
   If for each $i \in I$ we shift candidate~$p$ to the top of the preference
   order of voter $v_{S_i}$, then altogether we make $4h$~swaps and $p$~becomes
   a weak Condorcet winner. This is so because (i) $p$~is ranked on the fifth
   place in each of these votes, (ii) the swaps cause $p$ to pass $d$
   in~$h$~votes (so $p$ ties with $d$ in their head-to-head contest), and (iii)
   the swaps cause $p$ to pass each element candidate exactly once (so $p$ ties
   in a head-to-head contest with each element candidate). As a consequence,
   $\{p\}$ and~$\{d\}$ are two winning committees and we see that election
   result has changed.


   Let us now consider the opposite direction. We first note that if we perform
   up to $4h$~swaps, then we can change the winning margins indicated in
   Figure~\ref{fig:gehrlein-weakly-stable-reduction} by at most $8h$. As a
   consequence (and assuming that $m \geq 2$), after $4h$ swaps candidate $d$
   certainly is still preferred to each candidate other than $p$ by a majority
   of the voters. Further, after $4h$ swaps still at least half of the voters
   prefer $d$ to~$p$. This is so because in each vote either $p$ already is
   preferred to $d$ or it takes at least~four swaps to move $p$ ahead of $d$;
   this means that with $4h$ swaps, we can change at most $h$ voters to prefer
   $p$ over $d$ and this just enough to ensure that $p$ and $d$ tie in their
   head-to-head contest. As a consequence, after $4h$ swaps $d$ certainly is a
   (weak) Condorcet winner and $\{d\}$ is among the winning committees.

   To ensure that $\{d\}$ is not the only winning committee, it is
   necessary to guarantee that some other candidate is a weak
   Condorcet winner. Based
   on~Figure~\ref{fig:gehrlein-weakly-stable-reduction}, it is clear
   that after $4h$ swaps all element candidates and dummy candidates
   loose at least one head-to-head contest (assuming~$m>1$) and, so,
   only $p$ may become a weak Condorcet winner. For this to happen,
   (i) $p$ needs to pass $d$ in $h$ votes, and (ii) $p$ needs to pass
   each element candidate in at least one vote. A simple counting
   argument shows that this is possible only by shifting $p$ to the
   top position in $h$ votes from the first group that correspond to
   an exact cover of $X$ with $h$ triplets from $\mathcal{S}$.


   We conclude by noting that the reduction works in polynomial time.
\end{proof}

Without much surprise, we find that \textsc{Robustness Radius} is also
$\np$-hard for $\bordacc$ and STV. For these rules, however, the
hardness results are, in fact, significantly stronger. In both cases
it is already $\np$-hard to decide whether the outcome of the given
election changes after a single swap, and for STV the result holds
even for committees of size one ($\bordacc$ with committees of size
one is equivalent to the single-winner Borda rule, for which the
problem is polynomial-time solvable~\cite{shi-yu-elk:c:robustness}; this
also follows directly from Theorem~\ref{pro:weak-sep}).

\newcommand{\thmstvbcc}{\textsc{Robustness Radius} is $\np$-hard both
  for STV and for $\bordacc$, even if we can perform only a single
  swap; for STV this holds even for committees of size~$1$. For
  $\bordacc$, the problem 
  is $\wone$-hard with respect to the committee size.}


\begin{theorem}\label{thm:bcc-hard}
  $\beta$-CC \textsc{Robustness Radius} is $\np$-hard and $\wone$-hard
  with respect to the size of the committee even if the robustness
  radius is one.
\end{theorem}

\begin{proof}
  We show the result by giving a reduction from the \textsc{Regular
    Multicolored Independent Set} problem. In this problem we are given a
  regular graph $G$, where each vertex has degree $d$ and has one of
  $h$~colors, and we ask if there is an \emph{$h$-colored independent
    set}, that is, a size-$h$ set of pairwise non-adjacent vertices
  containing one vertex from each color class.  This problem is known
  to be both $\np$-complete and $\wone$-hard for the parameter
  $h$~\cite[Corollary 13.8]{CyganFKLMPPS15}.  To obtain our
  $\wone$-hardness result, we will ensure that the reduction uses
  committees of size that is a function of $h$ only (indeed, we will
  use committee size $h+2$; aside from this one restriction, we give a
  standard many-one reduction.\smallskip

  \noindent\textbf{Input Instance.}\quad
  Let ~$(G,h,d)$ be an instance of \textsc{Regular Multicolored
    Independent Set}. We let $s:=|V(G)|$ be the number of vertices in
  the input graph and $r:=|E(G)|$ be the number of its edges. We
  assume, without loss of generality, that $s \geq 2h$ (indeed, in a graph with no
  isolated vertices there is no independent set that contains more
  than half of the vertices). Below we describe the election that we
  use in our $\beta$-CC \textsc{Robustness Radius} instance.\smallskip

  \noindent\textbf{Candidates and Committee Size.}\quad
  The set of candidates consists of the vertex set~$V(G)$ of the
  graph~$G$, the set~$Z:=\{z_0,z_1,z_2\}$ of \emph{special
    candidates}, the set $X:=\{x_1,\dots,x_h\}$ of \emph{safe
    candidates}, and the set~$D$ of \emph{dummy candidates} (the
  number of dummy candidates and the fact that there are polynomially
  many of them with respect to $r+s$ will become clear later). We set
  the committee size~$k:=h+2$. \smallskip

  \noindent\textbf{High-Level Idea.}\quad
  The idea of the construction is to ensure that for our election
  the following holds:
  \begin{enumerate}
  \item The \emph{safe committee} $\{z_0,z_1,x_1,\dots,x_h\}$ is
    always winning (possibly uniquely).
  \item For each $V' \subseteq V$, if $V'$ is an $h$-colored
    independent set, then $\{z_0,z_2\} \cup V'$ is a winning
    committee.
  \item There are no other winning committees.
  \item Using a single swap of adjacent candidates---which gives the robustness
   radius of one---it is possible to ensure that the safe committee is the only
   winning committee (in other words, a single swap suffices to change the set
   of winning committees if and only if there is an $h$-colored independent set
   for $G$).
  \end{enumerate}
  In particular, we will ensure that if there is no $h$-colored
  independent set, then the safe committee will have dissatisfaction
  score lower by at least four points than the next best committee (so
  a single swap would not suffice to change the set of winning
  committees); for the notion of the dissatisfaction score, recall
  Remark~\ref{rem:dissat}.\smallskip



  \newcommand{\dissat}{\ensuremath{\Delta} }
  \newcommand{\dummies}{\ensuremath{\ \ggg\ }} \noindent\textbf{Dummy
    Candidates and the $\boldsymbol\dissat$ Value.}\quad We will
  ensure that the safe committee will have dissatisfaction score:
  \[
    \dissat :=8r + hs^2,
  \]
  and that, indeed, this will be the lowest possible dissatisfaction
  score (prior to performing swaps).  To simplify our construction, we
  use a number of dummy candidates and we adopt the following
  convention: Whenever we put some dummy candidate among the top~$\dissat$
  positions in a vote, we put this candidate beyond position $\dissat$
  in all other votes (on its own, this is not enough to guarantee
  that no dummy candidate belongs to a winning committee, but we will
  later show that this indeed is the case).  As a consequence, for $n$
  voters we need at most $O(n\dissat)$ dummy candidates. Since we
  will form only polynomially many voters, we will also need only
  polynomially many dummy candidates.\smallskip

 
  \noindent\textbf{Voters.}\quad In the following, we describe the voters of
  our election in four groups, each playing a specific role in the
  construction. We briefly mention the voters' respective roles and
  formally prove them later.  Whenever we put the symbol $\dummies$ in
  a preference order, we mean listing $\Delta$ ``fresh'' dummy
  candidates (i.e.,\,ones that are not ranked among the top $\dissat$
  positions by the other voters), followed by all the remaining
  candidates in some arbitrary order.
  
  \begin{description}
  \item[Special Candidate
    Voters.] 
    This group consists of $h+3$~voters with preference orders of the
    form $z_0 \succ \dummies$. These voters ensure that every winning
    committee includes candidate~$z_0$.
  \item[Safe Committee Voters.] For each color~$i \in [h]$, we form
    $(s+1)\cdot s/2+ 6d$~ voters with preference order
    $x_i \succ z_2 \succ \dummies$. These voters ensure that the safe
    committee~$\{z_0,z_1,x_1,\dots,x_h\}$ is indeed winning.
  \item[Vertex Selection Voters.] For each color~$i \in [h]$, we form
    $s$ voters, where each vertex candidate of color $i$ appears
    exactly once on each of the first $s$ positions, candidate $z_1$
    is ranked on the $(s+1)$-th position, and all other
    top~$\dissat$ positions are taken by the dummy candidates.
    Formally, we form these voters as follows. We start with $s$
    voters with preference orders:
  
    \begin{profile}{ccccccccccccc}
      v_1 &\succ& v_2 &\succ& \cdots &\succ& v_{s-1} &\succ& v_s     &\succ& z_1 &\succ& \dummies, \\
      v_2 &\succ& v_3 &\succ& \cdots &\succ& v_{s}   &\succ& v_1     &\succ& z_1 &\succ& \dummies, \\
      &\vdots&    &     &       &     &         &\vdots&        &     &     &     &           \\
      v_s &\succ& v_1 &\succ& \cdots &\succ& v_{s-2} &\succ& v_{s-1}
      &\succ& z_1 &\succ& \dummies.
    \end{profile}
    Then we replace each vertex candidate that is not of color~$i$
    with a fresh dummy candidate. The role of this group is to ensure
    that except for the safe committee, every other winning committee
    (if it exists) must contain exactly one vertex of each color.

  \item[Independent set voters.] For every edge~$\{u,v\}$ we introduce
    two pairs of voters, with preference orders of the form:
    \begin{align*}
      &u \succ v \succ z_0 \succ \dummies, \text{ and}\\
      &v \succ u \succ z_0 \succ \dummies.
    \end{align*}
    The role of this group is to ensure that if there is a winning
    committe that contains $h$ vertex candidates, then these vertices
    form an independent set.
  \end{description}
  This completes the construction. We note that it is computable in
  polynomial time. Before we formally prove the correctness of our
  construction, we discuss several important facts about possible
  winning committees for the constructed election.\smallskip
 
  \noindent\textbf{Safe Committee.} \quad
  First, observe that the safe committee $\{z_0,z_1,x_1,\dots,x_h\}$
  provides total dissatisfaction score equal to~$\dissat$. To see
  this, note that the special candidate voters and the safe committee
  voters have dissatisfaction score zero for it. For every color, the
  respective vertex selection voters together have dissatisfaction score
  equal to $s^2$. Thus, the dissatisfaction score of all vertex selection
  voters of all colors is $hs^2$. The independent set voters generate
  dissatisfaction score equal to $8r$ (for each edge, the two pairs of
  voters in total have dissatisfaction score~$8$). Altogether, the
  safe committee has dissatisfaction score
  $\dissat = 8r+ hs^2$.\smallskip

  \noindent\textbf{Independent Set Committees.} \quad
  Second, observe that every committee $\{z_0,z_2\} \cup V'$, where
  $V' \subseteq V(G)$ is an $h$-colored independent set, causes total
  dissatisfaction exactly~$\dissat$. Indeed, for such a committee the
  following holds (we provide additional explanations for the last two
  voter groups below):
  \begin{enumerate}
  \item Special candidate voters have dissatisfaction score equal to
    zero.
  \item Each safe committee voter has dissatisfaction score equal to one
    (due to candidate $z_2$), so altogether their dissatisfaction score is
    $h((s+1)\cdot s/2+ 6d) = (s+1)\cdot hs/2+ 6hd$.
   \item Vertex selection voters have total dissatisfaction score~$h(s-1)\cdot
    s/2$. To see this, consider a group of vertex selection voters for some
    color $i$. As $V'$ is $h$-colored, it contains exactly one vertex of color
    $i$, which these voters rank on all positions between $1$ and $s$ (and they
    rank all other committee members below these positions). This means that
    their dissatisfaction is $0+1+ \cdots + (s-1) = (s-1)\cdot {s}/{2}$. As
    there are $h$ colors, after multiplying this number by~$h$, we get our total
    dissatisfaction value.

   \item Independent set voters have total dissatisfaction score~$8r-6hd$. To
    see why this is the case, we first note that the voters in this group have
    dissatisfaction at most $8r$ due to
    candidate~$z_0$. However, for each edge $\{u,v\}$ such that $V'$
    contains exactly one of the vertex candidates $u$, $v$, this
    dissatisfaction is decreased by $6$ (if our committee contained
    both $u$ and $v$, then the dissatisfaction would be decreased by
    $8$, but this does not happen as we assumed $V'$ to be an
    independent set). Since our committee contains exactly $h$
    vertices and each vertex touches exactly $d$ unique edges (because
    $V'$ is an independent set), we have total dissatisfaction~$8r-6hd$.
  \end{enumerate}

  One can verify (and we will show this formally later) that if we
  replace $V'$ with a set of $h$~vertices of different colors that do
  not form an independent set, then the dissatisfaction would be
  higher by at least four points (intuitively, for every two points
  that we gain by ``covering'' some edge with two vertices rather than
  one, we lose six points for being able to cover one edge less).\smallskip

  \noindent\textbf{Losing Committees.}\quad  
  Next, we show that every other committee causes total
  dissatisfaction at least~$\dissat+4$. To this end, we distinguish
  between five cases for possible committees.

  \begin{description}  
  \item[Case 1 (committees that do not contain $\boldsymbol{z_0}$).]
    Every committee $C'$ that does not contain candidate $z_0$ causes
    total dissatisfaction at least~$\dissat+h$. When $z_0$~is not part
    of the committee, then up to $k=h+2$~voters from the special
    candidate voters group have dissatisfaction at least one (in best
    case, they are represented by their second-best choice), and the
    last one has dissatisfaction at least $\dissat$.  Thus $z_0$ must
    belong to all winning
    committees.

   \item[Case 2 (committees that contain $\boldsymbol{z_0}$, $\boldsymbol{z_1}$,
    and $\boldsymbol{z_2}$).] Every committee $C'$ that contains $z_0$,
    $z_1$, and $z_2$ causes total dissatisfaction at
    least~$\dissat+4$. To see this, let us first consider the
    dissatisfaction of the voters when they are represented
    by~$\{z_0,z_1,z_2\}$ only. In this case, the special candidate
    voters have zero dissatisfaction score, the safe committee voters have
    dissatisfaction score of $h((s+1)\cdot s/2+6d)$, the vertex selection
    voters have dissatisfaction score $hs^2$, and the independent set
    voters have dissatisfaction score $8r$. Thus the total dissatisfaction
    is:
    \[
      \big(h((s+1)\cdot s/2+6d)\big) + \big(hs^2\big) +\big(8r\big)
        = \Delta + h((s+1)\cdot s/2+6d).
    \]
    Let us now consider the remaining $h-1$ candidates. Each of the
    safe candidates can decrease the dissatisfaction by
    exactly $(s+1)\cdot s/2+6d$. Each of the vertex candidates can
    decrease the dissatisfaction by at most~$(s+1)\cdot s/2+6d$ (the
    first part comes from the vertex selection voters, who for a given
    vertex decrease the dissatisfaction by at most $1+2+ \cdots + s$,
    and the second one comes from the independent set
    voters\footnote{If an edge is covered by a single vertex
      candidate, the satisfaction decreases by $6$. If it is covered
      by two vertex candidates, it decreases by $8$, but we ``split''
      it over two candidates, so each of them decreases the
      dissatisfaction by~$4$.}).
    We have that
    $ h((s+1)\cdot s/2+6d) - (h-1)((s+1)\cdot s/2 + 6d) = (s+1)\cdot
    s/2 + 6d > 4$.  That is, altogether the remaining $h-1$ candidates
    cannot cause the dissatisfaction to be lower than $\Delta + 4$.
   
    
  \item[Case 3 (committees that contain $\boldsymbol{z_0}$ but not
    $\boldsymbol{z_2}$).] Consider a committee $C'$ that contains $z_0$
    and does not contain $z_2$. If it does not contain all 
    candidates from~$\{x_1,\dots,x_h\}$, then its dissatisfaction must
    be (much) larger than~$2\dissat$. For example, if it does not
    contain some given candidate $x_i$, then at least
    $(s+1)\cdot s/2+ 6d-(h+1)>2$ voters with preference order of the
    form $x_i \succ z_2 \succ \dummies$ are dissatisfied by at
    least~$\dissat$.  Thus let us assume that $C'$ contains~$z_0$ and
    all candidates from~$\{x_1, \ldots, x_h\}$. If it does
    not contain $z_1$, then---using similar reasoning as before---the
    vertex selection voters cause dissatisfaction (much) greater than
    $2\dissat$. In summary, the safe committee is the only committee
    that contains $z_0$, does not contain~$z_2$, and has
    dissatisfaction lower than $\dissat+4$ (indeed, as we have seen,
    it has dissatisfaction exactly~$\dissat$).

  \item[Case 4 (committees that contain $\boldsymbol{z_0}$ but not
    $\boldsymbol{z_1}$).]  Consider a committee $C'$ that contains
    $z_0$ and does not contain $z_1$.  If this committee does not
    contain at least a single vertex candidate for each color, then
    its dissatisfaction is (much) larger than~$2\dissat$. For example,
    let us assume that $C'$ does not contain vertex candidate of color
    $i$. Then, $s-(h+1)>1$~of the vertex selection voters
    corresponding to color~$i$ are dissatisfied by at least~$\dissat$.
    Thus let us assume that $C'$ contains at least one vertex
    candidate for each color. Then, if $C'$ does not contain~$z_2$,
    then it has dissatisfaction (much) greater than $2\dissat$ due to
    the safe committee voters.  In summary, if a committee contains
    $z_0$, does not contain $z_1$, and causes dissatisfaction lower
    than $\dissat+4$, then it must contain $z_2$ and a vertex
    candidate of each color.

  \item[Case 5 (non-independent set committees).] Finally, let $C'$ be a
    committee of the form $\{z_0,z_2\} \cup V'$, where $V'$ contains
    vertices for each color, but these vertices do not form an
    independent set. Such a committee causes dissatisfaction at
    least~$\dissat+4$.  The special candidate voters have
    dissatisfaction zero, the safe committee voters have
    dissatisfaction $h((s+1)\cdot s/2+6d)$, the vertex selection
    voters have dissatisfaction $h(s-1)\cdot s/2$, and the independent
    set voters have dissatisfaction at least least~$8r-6hd+4$. We have
    analyzed the dissatisfactions of the first three groups of voters
    when considering the independent set committees; the calculations
    are the same. Let us, thus, consider the final group of voters.
    Let~$q$ be the number of edges between vertices from~$V'$. There
    are $q$~edges that are covered twice (i.e.,\,by two vertices from~$V'$),
    $hd-2q$~edges that are covered once, and all remaining edges are uncovered.
    The total dissatisfaction of the independent
    set voters is at least $8r-6(hd-2q)-8q=8r-6hd+4q$. Since $V'$~is
    not an independent set, we have $q\ge1$ and the claim follows.

  \end{description}

  \noindent\textbf{Correctness of the Reduction.} \quad The correctness easily follows from the
  above discussion. On the one hand, if graph~$G$ does not contain an
  $h$-colored independent set, then the safe committee is the only
  winning committee with total dissatisfaction~$\dissat$ and every
  other committee has dissatisfaction at least~$\dissat+4$. Thus, a
  single swap cannot change the set of winning committees. On the
  other hand, if graph~$G$ does contain an $h$-colored independent
  set, then the safe committee is not a unique winning committee. It
  is easy to verify that then the safe committee does not win anymore
  if one swaps candidate~$z_2$ with some candidate~$x_i$ in some vote
  from the safe committee group.
\end{proof}

In fact, the proof of Theorem~\ref{thm:bcc-hard} implies much more
than stated in the theorem. In particular, our construction shows that
the problem remains $\np$-hard even if we are given the current
winning committee as part of the input.  Furthermore, the same
construction implies that deciding whether a given candidate belongs
to some $\bordacc$ winning committee is both $\np$-hard and
$\conp$-hard (the $\np$-hardness result is sometimes taken for granted
in the literature, but has not been shown formally yet; see, e.g.,
Footnote~4 in the work of Bredereck et
al.~\cite{bre-fal-nie-tal:c:multiwinner-sb}).
Formally, we consider the following problem.

\begin{definition}
  In the $\bordacc$ \textsc{Member} problem we are given an election $E = (C,V)$,
  a committee size~$k$, and a distinguished candidate $c^* \in C$.
  We ask whether candidate~$c^*$ belongs to some $\bordacc$ winning committee for
  election~$E$ and committee size~$k$.
\end{definition}

Regarding the $\bordacc$ \textsc{Member} problem, we obtain an even
stronger result than implied by Theorem~\ref{thm:bcc-hard} and we show
that it is $\theta^p_2$-complete (the proof of this result is
deferred to the appendix).
Inuitively, the class $\theta^p_2$ contains those problems that can be
solved in polynomial time, provided that one can ask polynomially-many
non-adaptive queries to an $\np$ oracle (by asking non-adaptive
queries, we mean that the algorithm first computes all the instances
of the $\np$ problems that it wants to have solved, and then receives
answers for all of them at the same time).  Problems that are
$\theta^p_2$-complete are---seemingly---harder than the $\np$-complete
ones, but easier than $\np^\np$-complete or $\conp^\np$-complete ones.
For more details on $\theta^p_2$ and many other complexity classes,
see, e.g., the textbook of Hemaspaandra and
Ogihara~\cite{hem-ogi:b:companion}.

\begin{theorem}\label{thm:bccmember}
  $\bordacc$ \textsc{Member} $\theta^p_2$-complete.
\end{theorem}

We conclude this section by showing that the \textsc{Robustness
  Radius} problem is $\np$-hard for STV, even if we consider its
single-winner variant (i.e.,\,if we fix the committee size to be $1$)
and consider exactly one swap.

\begin{theorem}\label{thm:stv-hard}
  STV \textsc{Robustness Radius} is $\np$-hard even for $k = 1$ and $B
  = 1$.
\end{theorem}
\begin{proof}
  We give a reduction from STV \textsc{Winner Determination}---the
  problem of deciding whether a given candidate is an STV winner in a
  given election. This problem is known to be
  $\np$-hard~\cite[Theorem~4]{con-rog-xia:c:mle} for the committee size~$k=1$.
  Let $I$ be an instance of the \textsc{STV Winner Determination} problem. In
  $I$ we are given an election $E = (C, V)$ with $n$ voters, and a distinguished
  candidate $c \in C$; we ask if there exists a valid run of STV such that $c$
  becomes a winner in $E$. Without loss of generality, we can assume that $c$ is
  ranked first by some voter.

  Based on $I$, we construct an instance $I'$ of the STV
  \textsc{Robustness Radius} problem as follows. We fix the new set of
  candidates to be $C' = C \cup \{d\}$; here $d$ is a dummy candidate
  needed by our construction. For each voter $v \in V$, we put $d$ in
  $v$'s preference ranking right behind $c$, and add two copies of
  such a modified vote to $I'$; we call such votes
  non-dummy. Additionally, we add $2n+1$~dummy voters who rank~$d$
  first, $c$ second, and all remaining candidates next (in some
  fixed arbitrary order). Candidate $d$ is the unique winner in this
  election as he or she is ranked first by the majority of the
  voters. If we want to change the outcome of the election with a
  single swap, then we definitely need to swap $c$ and $d$ in the
  preference order of one of the dummy voters (otherwise $d$ would
  still have the majority of first-place votes). Let us consider such
  a modified election and call it $E''$.

  Observe that if $c$ is a possible winner in $I$, then $c$ is also a
  possible winner in $E''$. Indeed, STV may first eliminate all the
  candidates except for $c$ and $d$. In such a truncated profile,
  there would be $2n+1$ voters who prefer $c$ to $d$ and $2n$ voters
  who prefer $d$ to $c$; hence $c$ would become a winner.

  If $c$ is not a possible winner in $I$, then $c$ will be eliminated
  before some other candidate from $C \cup \{d\}$ in every possible
  run of STV on $E''$. Indeed, in each sequence of eliminations
  performed by STV, either there will be a moment where $c$ is
  eliminated as one of several candidates with a given (lowest) number
  of first-place votes or there will be a moment when there are still
  some remaining candidates in $C \setminus \{c\}$ and each such
  candidate is ranked first by at least two more non-dummy voters than
  $c$; as a result each such candidate will be ranked first by more
  (dummy and non-dummy) voters than $c$. In particular, $c$ will be
  removed from the election before some candidate from
  $C \setminus \{c\}$, and, so, also before $d$. After $c$ is removed
  from $E''$, there will be at least $2n+1$ voters who rank $d$ first
  (recall that there is at least one voter in $E$ who ranks $c$ first
  and, so, there are at least two non-dummy voters who rank $c$ first
  and $d$ second) and, so, $d$ is the unique winner of the
  election. Consequently, we have shown that the outcome of election
  $E'$ can change with a single swap if and only if the answer to the
  original instance $I$ is ``yes.'' This completes the proof.
\end{proof}

\section{Parameterized Algorithms for the Robustness Radius Problem}\label{sec:fpt-algo}

We complement our discussion of the complexity of the
\textsc{Robustness Radius} problem by providing several $\fpt$ algorithms for it.
Recall that an $\fpt$ algorithm for a given parameter (e.g., the
number of candidates or the number of voters) is an algorithm whose
running time is of the form $f(\rho)|I|^{O(1)}$, where $\rho$~is the value of
the parameter and $|I|$ is the length of the encoding of the input
instance.

First, using the standard approach of formulating integer linear programs and
invoking the algorithm of Lenstra~\cite{len:j:integer-fixed}, we find
that \textsc{Robustness Radius} is in $\fpt$ when parameterized by the
number of candidates (the proof is implicit, e.g., in the works of
Dorn and Schlotter~\cite{dor-sch:j:parameterized-swap-bribery} and
Knop et al.~\cite{kno-kou-mni:c:nfold-bribery}).

\begin{proposition}
  \textsc{Robustness Radius} for $k$-Copeland, NED, STV, and $\bordacc$
  is in $\fpt$ when parameterized by the number of candidates.
\end{proposition}

\noindent
For STV and $\bordacc$ we have fixed-parameter tractability not only
with respect to the number $m$ of the candidates, as mentioned above,
but also with respect to the number $n$ of the voters. For the case of
STV, we assume that the committee size $k$ is such that we never need
to ``delete non-existent voters'' and we refer to committee sizes
where such deleting is not necessary as \emph{normal}. For example,
committee size $k$ is not normal if $k > n$ (where $n$ is the number
of voters). Another example is to take $n=12$ and $k=5$: We would need to
delete $q = \lfloor \frac{12}{5+1}\rfloor+1 = 3$ voters for each
committee member, which would require deleting ``$15$ voters out of $12$.''


\begin{theorem}
For normal committee sizes, STV \textsc{Robustness Radius} is in $\fpt$ when parameterized by the number~$n$ of the voters.
\end{theorem}

\begin{proof}
  Let $E = (C,V)$ be the input election and let $k$ be the size of the
  desired committee. Let~$n = |V|$ be the number of voters. Since
  $k$ is normal, we have that $k \leq n$. For each candidate~$c$, we
  define $\rank(c) := \min_{v \in V} (\pos_v(c))$, which we refer to
  as the \emph{rank of $c$} (intuitively, the rank of candidate~$c$ is the
  highest position on which $c$~appears in the profile).

  First, we prove that a candidate with a rank higher than $n$ cannot
  be a member of a winning committee. For the sake of contradiction,
  let us assume that there exists a candidate $c$ with $\rank(c) > n$
  who is a member of some winning committee $W$. When STV adds some
  candidate to the committee (this happens when the number of voters
  who rank such a candidate first matches or exceeds the quota
  $\lfloor \frac{n}{k+1}\rfloor +1$), it removes this candidate and at
  least one voter from the election. Thus, before $c$ were included in
  $W$, STV must have removed some candidate $c'$ from the election
  without adding it to $W$ (this is so because $c$ had to be ranked
  first by some voter to be included in the committee; for $c$ to be
  ranked first, STV had to delete at least $n$ candidates, so by the
  assumption that the committee size is normal, not all of them could
  have been included in the committee). Whenever STV eliminates a
  candidate, it always chooses one with the lowest Plurality score.
  Since at the moment when $c'$ was removed the Plurality score of $c$
  was equal to zero, we have that the Plurality score of $c'$ also
  must have been zero.  Consequently, removing $c'$ from the election
  did not affect the top preferences of the voters and, so, right
  after removing $c'$, STV removed another candidate with zero
  Plurality score. By repeating this argument sufficiently many times,
  we conclude that $c$ must have been eventually eliminated, and, so,
  could not have been added to~$W$. This gives a contradiction and
  proves our claim.

  Second, by analogous reasoning, we also conclude that the number of
  committees winning according to STV is bounded by a function of $n$:
  Let us analyze the first step of STV. Either there will be some
  candidate that meets the quota and STV will include him or her in
  the committee and it will  remove at least one of the voters while doing
  so, or none of the candidates will meet the quota. In the latter
  case, in the following steps STV will remove all candidates that are
  not ranked first by any voter. In the former case, it will repeat an
  analogous step. Eventually, after at most $n$ steps, it will either
  complete, or it will remove all but at most $n$ candidates. Then it
  will certainly finish within the next at most $n$ steps.
  As a consequence of this reasoning, one can also verify that there
  is an FPT algorithm (parameterized by the number of voters) that
  outputs all winning committees for a given STV election. Thus we can
  test in FPT time if a given sequence of swaps has led to changing
  the result of our election or not.

  Third, we observe that the robustness radius for our election is at
  most $n^2$. Indeed, we can take a member of a winning committee and
  with at most $n^2$ swaps we can push him or her to have rank $n+1$
  or higher.
  Such a candidate no longer belongs to any winning committee and, so,
  the outcome of the election is changed. From now on we focus on
  sequences of at most $n^2$ swaps.

  Fourth, we observe that in order to change the outcome of an
  election, we should only swap such pairs of candidates that at least
  one candidate in the pair has rank at most $n^2 + n$. Indeed,
  consider a candidate $c$ with $\rank(c) > n^2 + n$. After $n^2$
  swaps, the rank of this candidate would still be above $n$, so he or
  she still would not belong to any winning committee (indeed, as
  without the shifts, the candidate would be eliminated in the initial
  set of rounds, when the candidates with no first-place votes are
  eliminated). Thus, a swap of two candidates with ranks higher than
  $n^2 + n$ does not affect the set of winning committees (the exact
  positions of these two candidates have no influence on the STV
  outcome).

  As a result, it suffices to focus on the candidates with ranks at
  most $n^2 + n$. There are at most $n(n^2 + n)$ of them and,
  consequently, there are at most $(2n^3 + 2n^2)^{n^2}$ possible
  $n^2$-long sequences of swaps which we need to check in order to
  find the shortest one that guarantees the result change. For each
  sequence of swaps, we test in $\fpt$ time whether the election
  outcome changes.
%
This completes the proof. 
\end{proof}


The algorithm for the case of $\bordacc$ is more involved. Briefly
put, it relies on finding in $\fpt$ time (with respect to the number
of voters) either the unique winning committee or two committees tied
for victory. In the former case, it combines brute-force search with
dynamic programming, and in the latter case, either a single swap or a
greedy algorithm suffice. For clarity, we start with presenting the
first phase, that is, finding the unique winning committee or two tied
committees, as a separate proposition.

\begin{proposition}\label{pro:cc-unique-winner}
  There is an algorithm that runs in $\fpt$-time with respect to the
  number of voters and, given an election $E = (C,V)$ and a committee
  size~$k$, checks whether the election has a unique $\beta$-CC winning
  committee (in which case it outputs this committee) or whether there is
  more than one $\beta$-CC winning committee (in which case it outputs
  some two winning committees).
\end{proposition}
\begin{proof}
  Let $E = (C,V)$ be the input election and let $k$ be the committee
  size. Let~$n = |V|$ be the number of voters. If $k \geq n$, then
  every winning committee consists of each voter's most preferred
  candidate and sufficiently many other candidates to form a committee
  of size exactly $k$. In this case the algorithm can easily provide
  the required output, so we assume that $k < n$. To avoid trivial cases,
  without loss of generality, we also assume that there are more than
  $k$~candidates.

  Our algorithm proceeds by considering all partitions of $V$ into $k$
  disjoint sets (there are at most $k^n \leq n^n$ such partitions).
  For a partition $V_1, \ldots, V_k$ the algorithm proceeds as follows
  (intuitively, the voters in each group $V_i$ are to be represented
  by the Borda winner of the election $(C,V_i)$):
  \begin{enumerate}
  \item For each election $E_i = (C,V_i)$ we compute the set $B_i$ of
    candidates that are Borda winners of $E_i$.
  \item If each $B_i$ is a singleton and all $B_i$'s are distinct,
    then we store a single committee $W = B_1 \cup \cdots \cup B_k$.
    Otherwise, it is possible to form two distinct committees, $W_1$
    and $W_2$, such that for each $B_i$, $W_1 \cap B_i \neq \emptyset$
    and $W_2 \cap B_i \neq \emptyset$;\footnote{We can form $W_1$ and
      $W_2$ as follows. First, we form set $W_0$ by taking the union
      of all singletons among $B_1, \ldots, B_k$; we know that
      $|W_0| < k$ because otherwise we would not enter this part of
      the algorithm. Then we form a new sequence of sets
      $B'_1, \ldots, B'_t$ by removing from sequence
      $B_1, \ldots, B_k$ all those sets that have a nonempty
      intersection with $W_0$. If the new sequence turns out to be
      empty, then we form $W_1$ and $W_2$ by extending $W_0$ by adding
      arbitrary candidates, but so that $W_1$ and $W_2$ are distinct
      (it is possible because there are more than $k$ candidates in
      the election). If the new sequence is not empty, then we form
      $W_1$ and $W_2$ as follows: We include all members of $W_0$ in
      both sets and, then, for each $B'_i$ we include the
      lexicographically first member of $B'_i$ in $W_1$ and the
      lexicographically last one in~$W_2$. This ensures that~$W_1$ and
      $W_2$~are distinct. If they still contain fewer than $k$
      candidates, then we extend them by including arbitrary candidates
      (but so that they remain distinct; again this is possible
      because there are more than~$k$~candidates in the election).}
    we store both~$W_1$ and~$W_2$.
  \end{enumerate}
  We check if among the stored committees there is a unique committee
  $W$ such that every other stored committee has lower $\beta$-CC
  score. If such a committee exists, then we output it as the unique
  winning committee. Otherwise, there are two stored committees, $W_A$
  and $W_B$, that both have $\beta$-CC score greater than or equal to
  that of every other stored committee. We output $W_A$ and $W_B$ as
  two committees tied for winning (if there is more than one choice
  for $W_A$ and $W_B$, then we pick one pair arbitrarily). 
\end{proof}

Before we move on to the proof of the fixed-parameter tractability of~$\beta$-CC
\textsc{Robustness Radius}, we introduce some additional notation. Let $E =
(C,V)$ be some election and let $v$ be some voter in $V$. By $\topc(v)$ we mean
the candidate ranked first by $v$. By $\topc(E)$ we mean the set $\{ \topc(v)
\mid v \in V\}$, that is, the set of candidates that are ever ranked first in
election $E$. For a committee $W$, the representative of some voter $v$ is the
member of $W$ that $v$ ranks highest. Finally, for committee $W$ and voter $v$,
we define $\reppos_v(W)$ to be the position of $v$'s representative from $W$ in
$v$'s vote.

\newcommand{\thmbccfpt}{$\beta$-CC \textsc{Robustness Radius} is in
  $\fpt$ when parameterized by the number of voters.}

\begin{theorem}\label{thm:bcc-fpt}
  \thmbccfpt
\end{theorem}
\begin{proof}
  Let $E = (C,V)$ be the input election and let $k$ be the committee
  size. Let $m = |C|$ be the number of candidates. Using
  Proposition~\ref{pro:cc-unique-winner}, we check whether there is a
  unique $\beta$-CC winning committee in $E$. Depending on the result,
  we proceed by distinguishing whether there is a unique winning committee or
  nor.\smallskip

  \noindent{\bfseries There is a unique winning committee $\boldsymbol{W}$.}
  We first describe a function that encapsulates the effect of
  shifting forward a particular candidate within a given set of votes.
  For each voter $v$, each candidate $c$, and each nonnegative integer~$b$,
  we define $\shift(v,c,b)$ to be the vote obtained from that of
  $v$ by shifting $c$ by $b$ positions forward, and we define:
  \[
     g(v,c,b) = \beta_{m}(\pos_{\shift(v,c,b)}(c)) - \beta_m(\reppos_{\shift(v,c,b)}(W)).
  \]
  In other words, $g(v,c,b)$ is the difference between the Borda
  scores of $c$ and the highest-ranked member of $W$ in vote $v$ with
  $c$ shifted $b$ positions forward.

  Let $V'$ be some subset of voters, and let us rename the voters so
  that $V' = \{v_1, \ldots, v_{n'}\}$. For each candidate $c$ and each
  nonnegative integer $b$, we define:
  \[
    g(V',c,b) = \max\biggl\{ \sum_{i=1}^{n'}g(v_i,c,b_i) \,\biggl|\, b_1, \ldots, b_{n'} \geq 0 \text{ and } b_1 + \cdots + b_{n'} = b \biggr\}.
  \]
  Intuitively, $g(V',c,b)$ specifies how many points more $c$~would
  receive from the voters in $V'$ as their representative than these
  voters would assign to their representatives from $W$, if we shifted
  $c$ by $b$~positions forward in an optimal way. 

  We assume that $g(\emptyset, c,b) = 0$ for each choice of $c$ and
  $b$. We can compute $g(V',c,b)$ in polynomial time using dynamic
  programming and the following formula (for each $1 \leq t <
  n'$):\footnote{In fact, it is possible to compute $g(V',c,b)$ using
    a greedy algorithm, but the dynamic programming formulation is far
    easier and allows us to sidestep many special cases, such as what
    happens if $c$ is him or herself a member of $W$.}
  \begin{align*}
    g(\{v_1, \ldots, v_{t}\},c,b) = \max_{0 \leq b_{t} \leq b} g(\{v_1, \ldots, v_{t-1}\},c,b-b_{t}) + g(v_{t},c,b_{t}).
  \end{align*}

  With the function $g$ in hand, we are ready to describe the algorithm.
  We consider every partition of $V$ into $k$ disjoint subsets $V_1,
  \ldots, V_k$; let us fix one such partition. Our goal is to compute
  the smallest nonnegative integer $b$ such that there is 
  a sequence of nonnegative integers $b_1, \ldots, b_k$ that adds
  up to $b$, and a sequence
  $c_1, \ldots, c_k$ of (not necessarily distinct) candidates so that:
  \begin{enumerate}
  \item[(a)] $g(V_1,c_1,b_1) + \cdots + g(V_k,c_k,b_k) \geq 0$,
  \item[(b)] there is a committee $W'$ such that $\{c_1, \ldots, c_k\}
    \subseteq W'$ and $W' \neq W$.
  \end{enumerate}
  Intuitively, the role of candidates $c_1, \ldots, c_k$ is to be the
  representatives of the voters from the sets~$V_1, \ldots, V_k$,
  respectively, in a new committee $W'$, distinct from $W$, that
  either defeats $W$ or ties with it. More formally, condition~(a)
  ensures that there is a way to perform $b = b_1 + \cdots + b_k$
  swaps so that the score of committee $W'$ is at least as large as
  that of $W$, and condition (b) requires that $W' \neq W$ and deals
  with the possibility that candidates in $c_1, \ldots, c_k$ are not
  distinct.

  To compute $b$, we will need the following function ($C'$ is a
  subset of candidates---we will end up using only polynomially many
  different ones---$i$ is an integer in $[k]$, and $b$ is a
  nonnegative integer):
  \[
  f(C',i,b) = \max\biggl\{ \sum_{j=1}^i g(V_j,c_j,b_j) \,\biggl|\, 
  c_1, \ldots, c_i \in C', b_1, \ldots, b_i \geq 0, b_1 + \cdots + b_i = b\biggr\}.
  \]
  We have that the smallest value of $b$ such that $f(C',k,b) \geq 0$ is
  associated with candidates $c_1, \ldots, c_k$ and values $b_1,
  \ldots, b_k$ that satisfy condition (a) above, under the condition
  that $c_1, \ldots c_k$ belong to~$C'$. To obtain the smallest value
  of $b$ that is associated with values $b_1, \ldots, b_k$ and $c_1,
  \ldots, c_k$ that satisfy both conditions (a) and (b) above, it
  suffices to compute:
  \[
    b_{V_1, \ldots,V_k} = \min\{b \in \naturals \mid w \in W \land f(C-\{w\},k,b) \geq 0 \}.
  \]
  The fact that we use sets of the form $C - \{w\}$ in the invocation
  of function~$f$ ensures that we obtain committees distinct from $W$.
  The fact that we try all $w \in W$ guarantees that we try all
  possibilities. The smallest value~$b_{V_1,\ldots,V_k}$ over all
  the partitions of $V$ is the smallest number of swaps necessary to
  change the outcome of the election.

  It remains to show that we can compute function $f$ in polynomial
  time. This follows by assuming that $f(C',0,b) = 0$ (for each $C'$
  and $b$) and applying dynamic programming techniques on top of the
  following formula (which holds for each $i \in [k]$):
  \[
  f(C',i,b) = \max_{0 \leq b_i \leq b,\, c_i \in C'} f(C',i-1,b-b_i) + g(V_i,c_i,b_i).
  \]
  The part of the proof where there is a unique $\beta$-CC winning committee
  for~$E$ is complete.
  
  \smallskip
  
  \noindent{\bfseries There are at least two committees that tie for
    victory.} Let $W_A$ and $W_B$ be two $\beta$-CC winning committees
  for $E$ (the algorithm from Proposition~\ref{pro:cc-unique-winner}
  provides them readily). We check if there is some voter $v$ whose
  representatives under $W_A$ and $W_B$ are distinct. If such a voter
  exists, then a single swap is sufficient to prevent one of the
  committees from winning: Let $a$ be the representative of~$v$ under
  $W_A$, and let $b$ be the representative of $v$ under
  $W_B$. Without loss of generality, we assume that $a$ is ranked higher than $b$. It
  suffices to shift $b$ one position higher. It certainly is possible
  (since $b$ was ranked below $a$, he or she certainly is not ranked
  first) and it increases the $\beta$-CC score of $W_B$, while the
  score of $W_A$ either stays the same or decreases (the score of
  $W_A$ would stay the same, e.g., if $b$ were ranked just below $a$
  and $b$ also belonged to~$W_A$; candidate~$a$ certainly does not belong to
  $W_B$ because $v$ does not have $a$ as a representative
  under~$W_B$).  In consequence, $W_A$ certainly is not a winning
  committee after the swap and, thus, the set of winning committees
  changes.
  
  Let us now consider the case where each voter has the same
  representative under both~$W_A$ and~$W_B$, and let $R$ be the set of
  voters' representatives ($R \subseteq W_A \cap W_B$). Since $W_A$
  and $W_B$ are distinct, there are candidates $a \in W_A \setminus
  W_B$ and $b \in W_B \setminus W_A$ and, in consequence, we know that
  $|R| < k$. We claim that $R = \topc(E)$, that is, that each
  representative is ranked first by some voter. For the sake of
  contradiction, let us assume that there is a voter $v$ that is not
  represented by his or her top-preferred candidate. In this case,
  committee~$W_C$ obtained from~$W_A$ by replacing candidate~$a$ with
  candidate~$\topc(v)$ has a higher score than~$W_A$ (voter~$v$ has a
  higher-ranked representative and all other voters have the same
  or higher-ranked representatives), which contradicts the fact that~$W_A$~is a
  winning committee. Thus our claim holds.

  As a consequence, the $\beta$-CC winning committees for election $E$
  are exactly those that contain all candidates from $R$. To change the election
  outcome, we have to transform~$E$ to an election~$E'$ such that $\topc(E) \neq
  \topc(E')$. We consider two types of actions that achieve this effect:
  \begin{enumerate}
  \item Shift some candidate $c \in C \setminus R$ to the top position
    of some voter $v$, thus ensuring that for the resulting election
    $E'$ we have $c \in \topc(E')$ (and, by assumption, $c \notin
    \topc(E)$). 

  \item For some candidate $d \in R$ and each voter $v$ that ranks $d$
    on top, shift the top-ranked member of $R \setminus \{d\}$ to be
    ranked first. This creates election $E'$ such that $\topc(E')$ is
    strictly contained in $\topc(E)$.
  \end{enumerate}
  Actions of the first type include the cheapest one that creates an
  election $E'$ such that $\topc(E') \setminus \topc(E) \neq
  \emptyset$, and actions of the second type include the cheapest one
  that creates an election~$E'$ such that $\topc(E) \setminus
  \topc(E') \neq \emptyset$. Thus it suffices to compute the cheapest
  action of each type (there are only polynomially many actions to
  consider) and output its cost as the smallest number of swaps
  necessary to change the outcome of the election.  
\end{proof}

It is natural to ask whether the above theorem holds for other
variants of the Chamberlin--Courant rule (i.e., for variants based on
scoring functions other than the Borda one). This issue is quite
intriguing. While the first part of the proof---where we deal with
the case of a unique winning committee---is general and works for any
scoring function (indeed, it suffices to replace the Borda scoring
function $\beta$ in the definition of function $g$ with any other
scoring rule), the situation of the second part is harder to deal
with.  Indeed, in the second part of the proof, when we consider the
case where not all voters have the same representative, we rely on the
fact that a single swap of a representative will increase the score of
a committee. This is crucial for our argument, and due to this
assumption it does not matter which two specific  winning committees
$W_A$ and $W_B$ we obtained from
Proposition~\ref{pro:cc-unique-winner}.  Without it, we would have to
be more careful in choosing them.

We conclude this section by noting that the \textsc{Robustness Radius}
problem for $k$-Copeland and NED is\ $\wone$-hard for the
parameterization by the number of voters. This follows by a simple
adaptation of a $\wone$-hardness proof of Kaczmarczyk and
Faliszewski~\cite[Theorem~7]{kac-fal:c:destructive-shift-bribery} for
Copeland$^\alpha$ \textsc{Destructive Shift Bribery} (the idea of the
adaptation is to insert sufficiently many dummy candidates between the
non-dummy ones, so that the only reasonable swaps are those that shift
the designated candidate backward). Since the proof uses an odd number
of voters, it applies to NED as well.

\begin{corollary}
  \textsc{Robustness Radius} for $k$-Copeland and NED is $\wone$-hard
  when parameterized by the number of voters.
\end{corollary}

\section{Beyond the Worst Case: An Experimental Evaluation}\label{sec:exp-eval}


In this section we present results of experiments in which we measure
how many randomly-selected swaps are necessary to change election
results under our rules.\footnote{We omit NED because we found it to
  be computationally too expensive. However, we expect the results to
  be similar to the results that we have for $k$-Copeland$^\alpha$.}

We performed a series of experiments using five distributions of
rankings---three synthetic ones and two based on real-life datasets
obtained from the Pref\-Lib~\cite{mat-wal:c:preflib} library of
real-life preference data.  Regarding the real-life data, we used the
dataset of preferences over sushi sets~\cite{kam:c:sushi} and the
dataset with preferences over university courses (treating them as
distributions by selecting votes from them uniformly at
random). Regarding the synthetic distributions, we used the following
ones (see the description below or, for a more detailed discussion and
literature overview, a book chapter by Boutilier and
Rosenschein~\cite{bou-ros:incomplete-information}):
 \begin{enumerate}[(i)]
 \item Impartial Culture (IC),
 \item Mallows model with parameter $\phi$ between $0$ and $1$
  drawn uniformly at random, and
 \item a mixture of two Mallows models with two separate values of parameters
  $\phi_1$ and $\phi_2$ drawn uniformly and independently at random.
\end{enumerate}

In the Impartial Culture model, each preference order is drawn
uniformly at random.
In contrast, the intuition behind the Mallows model is that there is a given central
preference order and the more swaps are necessary to modify some preference
order~$r$ to become this central one, the less probable it is to
draw~$r$ (in particular, the central order is the most probable one to
be generated). Formally, the Mallows model consists of a
\emph{central order}~$r_0$ of $m$~elements and a dispersion
parameter~$\phi \in (0,1]$ which quantifies the concentration of the
rankings around the peak~$r_0$ with respect to some distance measure; we use
the Kendall tau distance~\cite{Kendall38}. In particular, the probability of generating a given
ranking~$r$~is:
$$P_{r_0,\phi}(r) = \frac{\phi^{d(r,r_0)}}{Z} \text{~~~where~~~}Z = 1
\cdot (1+\phi) \cdot (1+\phi+\phi^2) \cdots (1+\cdots+\phi^{m-1}),$$
and where~$d(r,r_0)$ is Kendall tau distance between $r$ and $r_0$,
that is, the number of swaps of adjacent candidates that are necessary
to transform $r$ into $r_0$. Note that the normalization constant~$Z$
is independent of~$r_0$. For $\phi = 1$, the Mallows model becomes
equivalent to the Impartial Culture model; for $\phi = 0$ it draws the
central ranking~$r_0$ only. In the mixture of two Mallows models, we
use models with different central orders and different values of the
dispersion parameter (both drawn independently and
uniformly at random). Additionally, we draw uniformly at random a value
$p \in [0,1]$ and for each vote that we are to generate, we use the
first model with probability $p$, and the second model with
probability $1-p$.

For each of our five distributions, and for each of the voting rules
that we consider,\footnote{For $k$-Copeland$^\alpha$ we took~$\alpha=0.5$.}
we performed $2000$~simulations. In each simulation we had drawn an election
containing $10$~candidates and $30$~voters from the given
distribution.
Then we were repeatedly drawing a pair of adjacent candidates uniformly at random and
performing a swap,
until the outcome of the election changed (in fact, we never did more
than~$5000$~swaps in order to change the outcome).  The average number
of swaps required to change the outcome of an election for different
rules and for different distributions is depicted in
Figure~\ref{fig:exp_plots}. We present the results for committee size
$k = 3$. We have also performed simulations for $k = 5$ that led to
analogous conclusions.  We note that the standard deviations in our
experiments were fairly high (usually close to the value of the
reported averages, but sometimes almost twice as large as the value of
the reported average).  This means that in many elections the required
number of random swaps was, in fact, notably smaller than the provided
average, and in some elections this number was significantly above the
average.


As expected, 
the robustness radius decreases with the increase of randomness in the
voters' preferences. Indeed, one needs relatively few swaps to change
the results of elections generated using the Impartial Culture
distribution, but changing the results of elections generated
according to the Mallows model requires many more (random) swaps.  It
is interesting that the results regarding the Mallows model are
somewhat different from those for the Sushi dataset, as it is often
believed that the Mallows model captures the preference orders from
the Sushi dataset well~\cite{kam:c:sushi}. Our results give some
circumstantial evidence that there is some nontrivial difference
between the Sushi dataset and the Mallows model (which, after all, is
to be expected---it is unlikely that a simple synthetic model would
capture real-life data perfectly). In particular, based on the fairly
small radiuses of the elections generated using the Sushi
distribution, we conclude that the preferences there are rather diverse.

Among our rules, 
$k$-Borda
is the most robust one ($k$-Copeland$^\alpha$, for $\alpha=0.5$, holds the second place),
whereas
%
rules that achieve either diversity ($\bordacc$ and, to some extent,
SNTV) or proportionality (STV)
are usually more vulnerable to small changes in the input. This is aligned
with what we have seen in the theoretical part of the paper (with a
minor exception of SNTV).
For the case of $k$-Borda, indeed, we would expect that many swaps
would cancel each other out (in terms of the effect on the Borda
scores of the candidates), which explains the rule's large
robustness. The performance of Borda can also be explained by noting
that it is a maximum likelihood estimator for a noise model that is
somewhat similar to ours (see, e.g., the overview provided by Elkind
and Slinko~\cite{elk-sli:b:rationalization}).


\begin{figure}[t!]
 \center
 \resizebox{\textwidth}{!}{
\begin{tikzpicture}

\begin{groupplot}[group style={group size=5 by 1, horizontal sep=.5cm}]
\nextgroupplot[
title={Mallows},
xmin=0, xmax=400,
ymin=-0.248333333333333, ymax=1.915,
xtick={0,100,200,300,400},
ytick={0,0.333333333333333,0.666666666666667,1,1.33333333333333,1.66666666666667},
yticklabels={$k$-Borda,$k$-Copeland$^{0.5}$,Bloc,$\beta$-CC,SNTV,STV},
minor xtick={},
minor ytick={},
tick align=outside,
tick pos=left,
x grid style={lightgray!92.02614379084967!black},
y grid style={lightgray!92.02614379084967!black}
]
\draw[fill=black,draw opacity=0] (axis cs:0,-0.15) rectangle (axis cs:354.255,0.15);
\draw[fill=lightgray!66.92810457516339!black,draw opacity=0] (axis cs:0,0.183333333333333) rectangle (axis cs:235.315,0.483333333333333);
\draw[fill=black,draw opacity=0] (axis cs:0,0.516666666666667) rectangle (axis cs:134.73,0.816666666666667);
\draw[fill=lightgray!66.92810457516339!black,draw opacity=0] (axis cs:0,0.85) rectangle (axis cs:149.82,1.15);
\draw[fill=black,draw opacity=0] (axis cs:0,1.18333333333333) rectangle (axis cs:55.705,1.48333333333333);
\draw[fill=lightgray!66.92810457516339!black,draw opacity=0] (axis cs:0,1.51666666666667) rectangle (axis cs:63.0454545455,1.81666666666667);
\nextgroupplot[
title={mixed Mallows},
xmin=0, xmax=400,
ymin=-0.248333333333333, ymax=1.915,
xtick={0,100,200,300,400},
ytick=\empty,
minor xtick={},
minor ytick={},
tick align=outside,
tick pos=left,
x grid style={lightgray!92.02614379084967!black},
y grid style={lightgray!92.02614379084967!black}
]
\draw[fill=black,draw opacity=0] (axis cs:0,-0.15) rectangle (axis cs:214.24,0.15);
\draw[fill=lightgray!66.92810457516339!black,draw opacity=0] (axis cs:0,0.183333333333333) rectangle (axis cs:136.705,0.483333333333333);
\draw[fill=black,draw opacity=0] (axis cs:0,0.516666666666667) rectangle (axis cs:104.155,0.816666666666667);
\draw[fill=lightgray!66.92810457516339!black,draw opacity=0] (axis cs:0,0.85) rectangle (axis cs:87.26,1.15);
\draw[fill=black,draw opacity=0] (axis cs:0,1.18333333333333) rectangle (axis cs:29.725,1.48333333333333);
\draw[fill=lightgray!66.92810457516339!black,draw opacity=0] (axis cs:0,1.51666666666667) rectangle (axis cs:38.07,1.81666666666667);
\nextgroupplot[
title={Sushi},
xmin=0, xmax=400,
ymin=-0.248333333333333, ymax=1.915,
xtick={0,100,200,300,400},
ytick=\empty,
minor xtick={},
minor ytick={},
tick align=outside,
tick pos=left,
x grid style={lightgray!92.02614379084967!black},
y grid style={lightgray!92.02614379084967!black}
]
\draw[fill=black,draw opacity=0] (axis cs:0,-0.15) rectangle (axis cs:184.815,0.15);
\draw[fill=lightgray!66.92810457516339!black,draw opacity=0] (axis cs:0,0.183333333333333) rectangle (axis cs:92.93,0.483333333333333);
\draw[fill=black,draw opacity=0] (axis cs:0,0.516666666666667) rectangle (axis cs:57.5,0.816666666666667);
\draw[fill=lightgray!66.92810457516339!black,draw opacity=0] (axis cs:0,0.85) rectangle (axis cs:56.335,1.15);
\draw[fill=black,draw opacity=0] (axis cs:0,1.18333333333333) rectangle (axis cs:23.09,1.48333333333333);
\draw[fill=lightgray!66.92810457516339!black,draw opacity=0] (axis cs:0,1.51666666666667) rectangle (axis cs:13.17,1.81666666666667);
\nextgroupplot[
title={University Courses},
xmin=0, xmax=400,
ymin=-0.248333333333333, ymax=1.915,
xtick={0,100,200,300,400},
ytick=\empty,
minor xtick={},
minor ytick={},
tick align=outside,
tick pos=left,
x grid style={lightgray!92.02614379084967!black},
y grid style={lightgray!92.02614379084967!black}
]
\draw[fill=black,draw opacity=0] (axis cs:0,-0.15) rectangle (axis cs:189.48,0.15);
\draw[fill=lightgray!66.92810457516339!black,draw opacity=0] (axis cs:0,0.183333333333333) rectangle (axis cs:118.19,0.483333333333333);
\draw[fill=black,draw opacity=0] (axis cs:0,0.516666666666667) rectangle (axis cs:117.96,0.816666666666667);
\draw[fill=lightgray!66.92810457516339!black,draw opacity=0] (axis cs:0,0.85) rectangle (axis cs:8.27,1.15);
\draw[fill=black,draw opacity=0] (axis cs:0,1.18333333333333) rectangle (axis cs:8.21,1.48333333333333);
\draw[fill=lightgray!66.92810457516339!black,draw opacity=0] (axis cs:0,1.51666666666667) rectangle (axis cs:9.26666666667,1.81666666666667);
\nextgroupplot[
title={Impartial Culture},
xmin=0, xmax=400,
ymin=-0.248333333333333, ymax=1.915,
xtick={0,100,200,300,400},
ytick=\empty,
minor xtick={},
minor ytick={},
tick align=outside,
tick pos=left,
x grid style={lightgray!92.02614379084967!black},
y grid style={lightgray!92.02614379084967!black}
]
\draw[fill=black,draw opacity=0] (axis cs:0,-0.15) rectangle (axis cs:36.285,0.15);
\draw[fill=lightgray!66.92810457516339!black,draw opacity=0] (axis cs:0,0.183333333333333) rectangle (axis cs:23.38,0.483333333333333);
\draw[fill=black,draw opacity=0] (axis cs:0,0.516666666666667) rectangle (axis cs:23.255,0.816666666666667);
\draw[fill=lightgray!66.92810457516339!black,draw opacity=0] (axis cs:0,0.85) rectangle (axis cs:21.815,1.15);
\draw[fill=black,draw opacity=0] (axis cs:0,1.18333333333333) rectangle (axis cs:14.59,1.48333333333333);
\draw[fill=lightgray!66.92810457516339!black,draw opacity=0] (axis cs:0,1.51666666666667) rectangle (axis cs:10.61,1.81666666666667);
\end{groupplot}

\end{tikzpicture}}
 \caption{\label{fig:exp_plots} Experimental results showing
   the average number of swaps needed to change the outcome of
   random elections obtained according to the description in
   Section~\ref{sec:exp-eval}. The standard deviations are quite high, 
   on the same order as the averages themselves (and often a bit larger).
 }
\end{figure}

The results for STV call for some additional discussion. Indeed, the
robustness radius of STV turned out to be close to~$10$ in the Sushi,
University Courses, and Impartial Culture distributions, whereas for
the Mallows model it was over~$60$, and for the mixture of two Mallows
models it was just below~$40$. The results for SNTV were qualitatively
similar, wheres $\bordacc$ typically achieve much higher robustness
radiuses (e.g.,\ in the Sushi dataset its average robustness
radius was more than four times larger than that of STV; for the other
datasets---except for the University Courses dataset---it was over two times
larger). This is not completely surprising as STV cannot be easily
interpreted as a maximum likelihood
estimator~\cite{con-san:c:likelihood-estimators,con-rog-xia:c:mle}
and, as per our Example~\ref{ex:correlated}, we should expect lower
robustness radiuses from rules focused on diversity and proportional
representation. Yet, the the fact that, on average, to change the
result of an election with 30 voters and 10 candidates (committee size~$3$) we
may need only about 10 random swaps of adjacent candidates is
worrisome. In many elections---especially in the low-stake and
medium-stake ones---we would expect many voters to make small
mistakes, where they rank two adjacent candidates in an opposite order
(e.g.,\ because these voters would be tired of the ranking process, or
because they would view these two candidates as similar etc.).  As a
consequence, for small STV elections there is a danger that the
outcome is affected by very minor, hard to predict, and hard to
observe issues. Since relatively small STV elections are common in
practice (e.g.,\ the rule is used by various universities and their
departments for internal elections), this result is quite
meaningful. In particular, the organizers of such elections may wish
to check if small numbers of random swaps can change the results of
their elections and, if so and if this is feasible, they might wish to
return to discussions on the voted issues (this would, of course,
require some agreement of the voters that if the outcome is not
``clear'' in the sense of the robustness radius, then the discussions
are resumed; this would be impossible in some settings, but would be
quite acceptable in others).

The above discussion is equally applicable to the case of SNTV, but
usually when SNTV elections are conducated, the voters only submit
their top preferences and, so, computing the robustness radius in the
sense of this section would be difficult. For the case of $\bordacc$,
the test could be executed---and might be meaningful and
reasonable---but the danger of non-robust results seems to be smaller
than in the case of STV (yet, note that for the University Courses
dataset the results of $\bordacc$ are as non-robust as those of STV).

\section{Conclusions}\label{sec:conclusions}

We formalized the notion of robustness of multiwinner rules and
studied the complexity of assessing the robustness/confidence of
collective multiwinner decisions. Our theoretical and experimental
analysis indicates that $k$-Borda is the most robust among our rules,
and that proportional rules, such as STV and the Chamberlin--Courant
rule, are on the other end of the spectrum. Indeed, for these rules we
suggest that organizers of small-scale elections run tests of the
robustness of the obtained results.

Our notions of robustness have already attracted attention of other
researchers, who have, for example, studied the complexity of the
\textsc{Robustness Radius} problem for the Chamberlin--Courant rule in
more detail~\cite{mis-son:c:robustness} (e.g., by considering
structured preference profiles) or who have considered the approval
setting~\cite{mis-son:c:robustness,gaw-fal:c:robustness}.  Other
interesting research directions involve analyzing the robustness
levels of multiwinner rules in the restricted preference domains
(e.g.,\,single-peaked preferences or single-crossing preferences),
considering counting variants of our problems to assess the
probability that a given number of random swaps can change the results
(see the initial results of Gawron and
Faliszewski~\cite{gaw-fal:c:robustness}), and finding natural voting
rules with robustness levels strictly between $1$ and $k$.  A more
open-ended research direction is to seek further notions of
robustness, both for the single- and multi-winner voting settings.

\bigskip

\noindent
\textbf{Acknowledgments.}\quad We are grateful to the anonymous \emph{SAGT 2017}
reviewers for their useful comments.
Robert Bredereck was partially supported by the DFG fellowship BR 5207/2.
Piotr Faliszewski was supported by the National Science Centre,
Poland, under project 2016/21/B/ST6/01509.
Andrzej Kaczmarczyk was supported by the DFG project AFFA (BR 5207/1
and NI 369/15).
Piotr Skowron was supported by a Humboldt Research Fellowship for Postdoctoral
Researchers (Alexander von~Humboldt Foundation, Bonn) while staying at
TU~Berlin.
Nimrod Talmon was supported by an I-CORE ALGO fellowship.

\bibliographystyle{plain}
\bibliography{tailored_grypiotr.bib}

\begin{thebibliography}{10}

\bibitem{azi-elk-fal-lac-sko:c:multiwinner-condorcet}
H.~Aziz, E.~Elkind, P.~Faliszewski, M.~Lackner, and P.~Skowron.
\newblock The {Condorcet} principle for multiwinner elections: {From}
  shortlisting to proportionality.
\newblock In {\em Proceedings of the 26th International Joint Conference on
  Artificial Intelligence}, pages 84--90, 2017.

\bibitem{bar-coe:j:non-controversial-k-names}
S.~Barber\`a and D.~Coelho.
\newblock How to choose a non-controversial list with $k$ names.
\newblock {\em Social Choice and Welfare}, 31(1):79--96, 2008.

\bibitem{bet-sli-uhl:j:mon-cc}
N.~Betzler, A.~Slinko, and J.~Uhlmann.
\newblock On the computation of fully proportional representation.
\newblock {\em Journal of Artificial Intelligence Research}, 47(1):475--519,
  2013.

\bibitem{blom2017towards}
M.~Blom, P.~J. Stuckey, and V.~Teague.
\newblock Toward computing the margin of victory in {S}ingle {T}ransferable
  {V}ote elections.
\newblock {\em INFORMS Journal on Computing}, 2019.
\newblock Published Online.

\bibitem{bou-ros:incomplete-information}
C.~Boutilier and J.~S. Rosenschein.
\newblock Incomplete information and communication in voting.
\newblock In F.~Brandt, V.~Conitzer, U.~Endriss, J.~Lang, and A.~D. Procaccia,
  editors, {\em Handbook of Computational Social Choice}, pages 223--257.
  Cambridge University Press, 2016.

\bibitem{BFKNST17}
R.~Bredereck, P.~Faliszewski, A.~Kaczmarczyk, R.~Niedermeier, P.~Skowron, and
  N.~Talmon.
\newblock Robustness among multiwinner voting rules.
\newblock In {\em Proceedings of the 10th International Symposium on
  Algorithmic Game Theory}, pages 80--92, 2017.

\bibitem{bre-fal-nie-tal:c:multiwinner-sb}
R.~Bredereck, P.~Faliszewski, R.~Niedermeier, and N.~Talmon.
\newblock Complexity of shift bribery in committee elections.
\newblock In {\em Proceedings of the 30th AAAI Conference on Artificial
  Intelligence}, pages 2452--2458, 2016.

\bibitem{car-hem-hem:b:dodgson-young}
I.~Caragiannis, E.~Hemaspaandra, and L.~Hemaspaandra.
\newblock Dodgson's rule and {Young's} rule.
\newblock In F.~Brandt, V.~Conitzer, U.~Endriss, J.~Lang, and A.~D. Procaccia,
  editors, {\em Handbook of Computational Social Choice}, pages 103--126.
  Cambridge University Press, 2016.

\bibitem{car:c:margin-of-victory}
D.~Cary.
\newblock Estimating the margin of victory for instant-runoff voting.
\newblock Presented at 2011 Electronic Voting Technology Workshop/Workshop on
  Trustworthy Elections, 2011.

\bibitem{cha-cou:j:cc}
J.~Chamberlin and P.~Courant.
\newblock Representative deliberations and representative decisions:
  {Proportional} representation and the {B}orda rule.
\newblock {\em American Political Science Review}, 77(3):718--733, 1983.

\bibitem{coelho:thesis:understanding}
D.~Coelho.
\newblock {\em Understanding, Evaluating and Selecting Voting Rules Through
  Games and Axioms}.
\newblock PhD thesis, Universitat Aut{\`o}noma de Barcelona, 2004.

\bibitem{con-rog-xia:c:mle}
V.~Conitzer, M.~Rognlie, and L.~Xia.
\newblock Preference functions that score rankings and maximum likelihood
  estimation.
\newblock In {\em Proceedings of the 21st International Joint Conference on
  Artificial Intelligence}, pages 109--115, 2009.

\bibitem{con-san:c:likelihood-estimators}
V.~Conitzer and T.~Sandholm.
\newblock Common voting rules as maximum likelihood estimators.
\newblock In {\em Proceedings of the 21st Conference in Uncertainty in
  Artificial Intelligence}, pages 145--152, July 2005.

\bibitem{CyganFKLMPPS15}
M.~Cygan, F.~Fomin, L.~Kowalik, D.~Lokshtanov, D.~Marx, M.~Pilipczuk,
  M.~Pilipczuk, and S.~Saurabh.
\newblock {\em Parameterized Algorithms}.
\newblock Springer, 2015.

\bibitem{dey-nar:c:sampling-margin-of-victory}
P.~Dey and Y.~Narahari.
\newblock Estimating the margin of victory of an election using sampling.
\newblock In {\em Proceedings of the 24th International Joint Conference on
  Artificial Intelligence}, pages 1120--1126, 2015.

\bibitem{dor-sch:j:parameterized-swap-bribery}
B.~Dorn and I.~Schlotter.
\newblock Multivariate complexity analysis of swap bribery.
\newblock {\em Algorithmica}, 64(1):126--151, 2012.

\bibitem{DF13}
R.~G. Downey and M.~R. Fellows.
\newblock {\em Fundamentals of Parameterized Complexity}.
\newblock Springer, 2013.

\bibitem{Droop1881}
H.~R. Droop.
\newblock On methods of electing representatives.
\newblock {\em Journal of the Statistical Society of London}, 44(2):141--202,
  1881.

\bibitem{elk-fal-las-sko-sli-tal:c:2d-multiwinner}
E.~Elkind, P.~Faliszewski, J.~Laslier, P.~Skowron, A.~Slinko, and N.~Talmon.
\newblock What do multiwinner voting rules do? {An} experiment over the
  two-dimensional euclidean domain.
\newblock In {\em Proceedings of the 31st AAAI Conference on Artificial
  Intelligence}, pages 494--501, 2017.

\bibitem{elk-fal-sko-sli:j:multiwinner-properties}
E.~Elkind, P.~Faliszewski, P.~Skowron, and A.~Slinko.
\newblock Properties of multiwinner voting rules.
\newblock {\em Social Choice and Welfare}, 48(3):599--632, 2017.

\bibitem{elk-fal-sli:c:swap-bribery}
E.~Elkind, P.~Faliszewski, and A.~Slinko.
\newblock Swap bribery.
\newblock In {\em Proceedings of the 2nd International Symposium on Algorithmic
  Game Theory}, pages 299--310, 2009.

\bibitem{elk-sli:b:rationalization}
E.~Elkind and A.~Slinko.
\newblock Rationalizations of voting rules.
\newblock In F.~Brandt, V.~Conitzer, U.~Endriss, J.~Lang, and A.~D. Procaccia,
  editors, {\em Handbook of Computational Social Choice}, chapter~8. Cambridge
  University Press, 2016.

\bibitem{fal-sko-sli-tal:b:multiwinner-voting}
P.~Faliszewski, P.~Skowron, A.~Slinko, and N.~Talmon.
\newblock Multiwinner voting: {A} new challenge for social choice theory.
\newblock In U.~Endriss, editor, {\em Trends in Computational Social Choice},
  pages 27--47. AI Access Foundation, 2017.

\bibitem{fal-sko-sli-tal-tal:j:hierarchy-committee}
P.~Faliszewski, P.~Skowron, A.~Slinko, and N.~Talmon.
\newblock Committee scoring rules: {Axiomatic} characterization and hierarchy.
\newblock {\em ACM Transactions on Economics and Computation}, 6(1):Article~3,
  2019.

\bibitem{fal-sli-sta-tal:j:clustering}
P.~Faliszewski, A.~Slinko, K.~Stahl, and N.~Talmon.
\newblock Achieving fully proportional representation by clustering voters.
\newblock {\em Journal of Heuristics}, 24(5):725--756, 2018.

\bibitem{fil-tal:c:sampling-monitoring}
A.~Filtser and N.~Talmon.
\newblock Distributed monitoring of election winners.
\newblock In {\em Proceedings of the 16th International Conference on
  Autonomous Agents and Multiagent Systems}, pages 1160--1168, 2017.

\bibitem{fit-hem-hoo-nar:t:hard-control}
Z.~Fitzsimmons, E.~Hemaspaandra, A.~Hoover, and D.~Narv{\'{a}}ez.
\newblock Very hard electoral control problems.
\newblock In {\em Proceedings of the 33rd AAAI Conference on Artificial
  Intelligence}, 2019.
\newblock To appear.

\bibitem{FG06}
J.~Flum and M.~Grohe.
\newblock {\em Parameterized Complexity Theory}.
\newblock Springer, 2006.

\bibitem{gar-joh:b:int}
M.~Garey and D.~Johnson.
\newblock {\em Computers and Intractability: {A} Guide to the Theory of
  {NP}-Completeness}.
\newblock {W. H. Freeman and Company}, 1979.

\bibitem{gaw-fal:c:robustness}
G.~Gawron and P.~Faliszewski.
\newblock Robustness of approval-based multiwinner voting rules.
\newblock In {\em Proceedings of the 6th International Conference on
  Algorithmic Decision Theory}, 2019.
\newblock To appear.

\bibitem{geh:j:multiwinner-condorcet}
W.~Gehrlein.
\newblock The {C}ondorcet criterion and committee selection.
\newblock {\em Mathematical Social Sciences}, 10(3):199--209, 1985.

\bibitem{gon85}
Teofilo~F. Gonzalez.
\newblock Clustering to minimize the maximum intercluster distance.
\newblock {\em Theoretical Computer Science}, 38:293--306, 1985.

\bibitem{hem-spa-vog:j:kemeny}
E.~Hemaspaandra, H.~Spakowski, and J.~Vogel.
\newblock The complexity of {Kemeny} elections.
\newblock {\em Theoretical Computer Science}, 349(3):382--391, 2005.

\bibitem{hem-ogi:b:companion}
L.~Hemaspaandra and M.~Ogihara.
\newblock {\em The Complexity Theory Companion}.
\newblock Springer, 2002.

\bibitem{kac-fal:c:destructive-shift-bribery}
A.~Kaczmarczyk and P.~Faliszewski.
\newblock Algorithms for destructive shift bribery.
\newblock {\em Autonomous Agents and Multi-Agent Systems}, 33(3):275--297,
  2019.

\bibitem{kam:c:sushi}
T.~Kamishima.
\newblock Nantonac collaborative filtering: {Recommendation} based on order
  responses.
\newblock In {\em Proceedings of the 9th International Conference on Knowledge
  Discovery and Data Mining}, pages 583--588, 2003.

\bibitem{kam:j:stable-rules-again}
E.~Kamwa.
\newblock On stable voting rules for selecting committees.
\newblock {\em Journal of Mathematical Economics}, 70:36--44, 2017.

\bibitem{Kendall38}
M.~G. Kendall.
\newblock {A new Measure of Rank Correlation}.
\newblock {\em Biometrika}, 30(1-2):81--93, 1938.

\bibitem{kno-kou-mni:c:nfold-bribery}
D.~Knop, M.~Kouteck{\'{y}}, and M.~Mnich.
\newblock Voting and bribing in single-exponential time.
\newblock In {\em Proceedings of the 34th Annual Symposium on Theoretical
  Aspects of Computer Science}, pages 46:1--46:14, 2017.

\bibitem{koc-kol-elk-fal-tal:c:multigoal}
M.~Kocot, A.~Kolonko, E.~Elkind, P.~Faliszewski, and N.~Talmon.
\newblock Multigoal committee selection.
\newblock In {\em Proceedings of the 28th International Joint Conference on
  Artificial Intelligence}, 2019.
\newblock To appear.

\bibitem{len:j:integer-fixed}
H.~{Lenstra, Jr.}
\newblock Integer programming with a fixed number of variables.
\newblock {\em Mathematics of Operations Research}, 8(4):538--548, 1983.

\bibitem{bou-lu:c:chamberlin-courant}
T.~Lu and C.~Boutilier.
\newblock Budgeted social choice: From consensus to personalized decision
  making.
\newblock In {\em Proceedings of the 22nd International Joint Conference on
  Artificial Intelligence}, pages 280--286, 2011.

\bibitem{mag-riv-she-wag:c:stv-bribery}
T.~Magrino, R.~Rivest, E.~Shen, and D.~Wagner.
\newblock Computing the margin of victory in {IRV} elections.
\newblock Presented at 2011 Electronic Voting Technology Workshop/Workshop on
  Trustworthy Elections, 2011.

\bibitem{mat-wal:c:preflib}
N.~Mattei and T.~Walsh.
\newblock Preflib: A library for preferences.
\newblock In {\em Proceedings of the 3nd International Conference on
  Algorithmic Decision Theory}, pages 259--270, 2013.

\bibitem{mcg:j:election-graph}
D.~McGarvey.
\newblock A theorem on the construction of voting paradoxes.
\newblock {\em Econometrica}, 21(4):608--610, 1953.

\bibitem{mis-son:c:robustness}
N.~Misra and C.~Sonar.
\newblock Robustness radius for {Chamberlin-Courant} on restricted domains.
\newblock In {\em Proceedings of the 45th International Conference on Current
  Trends in Theory and Practice of Computer Science}, pages 341--353, 2019.

\bibitem{nie:b:invitation-fpt}
R.~Niedermeier.
\newblock {\em Invitation to Fixed-Parameter Algorithms}.
\newblock Oxford University Press, 2006.

\bibitem{pet:c:total-unimodularity}
D.~Peters.
\newblock Single-peakedness and total unimodularity: {New} polynomial-time
  algorithms for multi-winner elections.
\newblock In {\em Proceedings of the 32nd AAAI Conference on Artificial
  Intelligence}, pages 1169--1176, 2018.

\bibitem{pro-ros-zoh:j:proportional-representation}
A.~Procaccia, J.~Rosenschein, and A.~Zohar.
\newblock On the complexity of achieving proportional representation.
\newblock {\em Social Choice and Welfare}, 30(3):353--362, 2008.

\bibitem{sek-sik-xia:c:bundling-condorcet}
S.~Sekar, S.~Sikdar, and L.~Xia.
\newblock {Condorcet} consistent bundling with social choice.
\newblock In {\em Proceedings of the 16th International Conference on
  Autonomous Agents and Multiagent Systems}, pages 33--41, 2017.

\bibitem{shi-yu-elk:c:robustness}
D.~Shiryaev, L.~Yu, and E.~Elkind.
\newblock On elections with robust winners.
\newblock In {\em Proceedings of the 12th International Conference on
  Autonomous Agents and Multiagent Systems}, pages 415--422, 2013.

\bibitem{sko-fal-sli:j:multiwinner}
P.~Skowron, P.~Faliszewski, and A.~Slinko.
\newblock Achieving fully proportional representation: {Approximability}
  results.
\newblock {\em Artificial Intelligence}, 222:67--103, 2015.

\bibitem{sko-fal-sli:j:axiomatic-committee}
P.~Skowron, P.~Faliszewski, and A.~Slinko.
\newblock Axiomatic characterization of committee scoring rules.
\newblock {\em Journal of Economic Theory}, 180:244--273, 2019.

\bibitem{sko-yu-fal-elk:j:sc-cc}
P.~Skowron, L.~Yu, P.~Faliszewski, and E.~Elkind.
\newblock The complexity of fully proportional representation for
  single-crossing electorates.
\newblock {\em Theoretical Computer Science}, 569:43--57, 2015.

\bibitem{xia:margin-of-victory}
L.~Xia.
\newblock Computing the margin of victory for various voting rules.
\newblock In {\em Proceedings of the 13th ACM Conference on Electronic
  Commerce}, pages 982--999, 2012.

\end{thebibliography}

\appendix

\section{$\boldsymbol{\theta^p_2}$-Hardness of Testing Membership in a Winning
$\boldsymbol\beta$-CC Committee}

In this section we show that the $\bordacc$ \textsc{Member} problem is
$\theta^p_2$-complete. To show $\theta^p_2$-membership, we define two
auxiliary $\np$-problems, $Q_1$ and~$Q_2$:
 \begin{description}
 \item[Problem $\boldsymbol{Q_1}$:] Given an election~$(C,V)$ and an
   integer~$r$, in $Q_1$ we ask if there is a committee that has
   $\bordacc$-score \emph{greater than}~$r$.
 \item[Problem $\boldsymbol{Q_2}$:] Given an election~$(C,V)$, a
   distinguished candidate~$c^*$, and an integer~$r$, in $Q_2$ we ask
   if there is a committee that contains candidate~$c^*$ and has
   $\bordacc$-score \emph{at least}~$r$.
 \end{description}
 Note that $Q_1$ is in~$\np$ because a committee with $\bordacc$-score
 at least~$r+1$ is a polynomial-size certificate for a
 ``yes''-instance.  Analogously, a committee containing~$c^*$, with
 $\bordacc$-score at least~$r$ is a polynomial-size certificate for a
 ``yes''-instance of~$Q_2$.

 A given candidate~$c^*$ belongs to some $\bordacc$ winning committee
 for some election~$(C,V)$ if and only if there is
 some~$r \in [0,|V|\cdot(|C|-1)]$ such that $((C,V),r)$ is a
 ``no''-instance of~$Q_1$ and~$((C,V),c^*,r)$ is a ``yes''-instance
 of~$Q_2$. This can be checked by a deterministic Turing machine that
 asks $2\cdot(|V|\cdot(|C|-1)+1)$ non-adaptive queries to an
 $\np$-oracle (as required by the definition of $\theta^p_2$; see,
 e.g.,\,the textbook of Hemaspaandra and
 Ogihara~\cite{hem-ogi:b:companion}).  Thus, $\bordacc$
 \textsc{Member} is in~$\theta^p_2$.

 Inspired by the work of Fitzsimmons et
 al.~\cite{fit-hem-hoo-nar:t:hard-control}, we establish
 $\theta^p_2$-hardness using the polynomial-time many-one reduction
 from \textsc{Vertex Cover Member}~\cite{hem-spa-vog:j:kemeny}.  In
 this problem we are given an undirected graph~$G=(V(G), E(G))$ and a
 distinguished vertex~$v^*$, and we ask if there is a minimum-size
 vertex cover $V' \subseteq V(G)$ that contains~$v^*$. Hemaspaandra et
 al.~\cite{hem-spa-vog:j:kemeny} showed that \textsc{Vertex Cover Member}
 is~$\theta^p_2$-hard.

 To simplify our proof, we show that $\theta^p_2$-hardness holds even if the
 input graph is regular.
 
 \begin{lemma}
  \textsc{Vertex Cover Member} is $\theta^p_2$-complete, even if the input graph is regular.
 \end{lemma}
 
 \begin{proof}
   Given a graph~$G=(V(G), E(G))$ and a distinguished
   vertex~$v^* \in V(G)$, we extend it to a new graph~$G'$ such that
   every vertex in~$G'$ has the same degree and vertex~$v^*$ is part
   of some minimum-size vertex cover in~$G'$ if and only if it is also
   part of some minimum-size vertex cover in~$G$.
  
   Let~$d=|E(G)|$ denote the desired, common degree of the vertices
   in~$G'$ (without loss of generality we assume that $G$ is connected
   and is not a tree, so $|E(G)| \geq |V(G)|$; we also assume that
   $d > 6$). We note that prior to adding vertices and edges to $G$
   (to form $G'$), each vertex of $G$ has degree at most $d$. We will
   form $G'$ by introducing some number of new vertices and some
   edges; each new edge will either connect two new vertices or one
   new vertex and one original vertex. The sum of the degrees of the
   vertices in $G$ is $2|E(G)|$, but if each of these vertices were to
   have degree~$d$, then this sum would be $d\cdot |V(G)|$. As a
   consequence, we need to add:
   \[
     d\cdot|V(G)|-2|E(G)| = d\cdot(V(G)-2)
   \]
   edges that connect original vertices with the new ones. We set
   $t = |V(G)|-2$ and we form $t$ \emph{degree-filling gadgets} such
   that each gadget provides $d$ edges between the old and the new
   vertices. Each degree-filling gadget is constructed as follows: We
   have two sets of new vertices, $A$ and~$B$, with
   $A:=\{a_1,\dots,a_d\}$ and~$B:=\{b_1,\dots,b_{d-3}\}$.  Every
   vertex from~$B$ is connected with every vertex from~$A$ (these are
   the only edges that touch vertices from $B$ in the gadget).
   Vertices from~$A$ are connected in a cyclic way, so that there is an
   edge between $a_1$ and $a_2$, between $a_2$ and $a_3$, and so on,
   until the edge between $a_d$ and $a_1$.  Moreover, each vertex
   from~$A$ is connected to a single original vertex (in an arbitrary
   way, but ensuring that, after considering all the degree-filling
   gadgets, every original vertex has degree~$d$). Note that
   graph~$G'$ indeed contains only vertices of degree exactly~$d$.

   Let us now consider some degree-filling gadget and its minimum-size
   vertex cover. We claim that this vertex cover contains exactly $d$
   vertices. Indeed, to cover the cycle between the vertices from~$A$,
   the cover needs to include at least $d/2$ vertices
   from~$A$. Further, the cover either needs to include all
   vertices from $A$ or all vertices from $B$ (otherwise some edge
   connecting a vertex from $A$ with a vertex from $B$ would not be
   covered). By including all vertices from~$A$ we get a cover of
   size $d$, whereas by including all vertices from $B$ we get a
   cover of size at least $d/2+d-3$ (which is greater than $d$,
   provided that $d > 6$, as assumed). Thus, without loss of
   generality, we can assume that each minimum-size vertex cover of
   $G'$ uses exactly $d$ vertices (of type $A$) from each degree-filling gadget.

  

   Consequently, there is a minimum-size vertex cover, say~$S$,
   for~$G'$ that contains all vertices of type $A$ from all degree-filling
   gadgets. These vertices cover all edges that were not
   originally in $G$. Hence, the remaining vertices in~$S$ come
   from~$V(G)$ and form a minimum vertex cover of~$G$.
 \end{proof}

 Now we are ready to show $\theta^p_2$-hardness of $\bordacc$
 \textsc{Member}, that is, we provide Theorem~\ref{thm:bccmember}.

 \paragraph{Construction Idea and Candidates.}
 We give a reduction from the \textsc{Vertex Cover Member} problem for
 regular graphs to the $\bordacc$ \textsc{Member} problem.
 Let~$G=(V(G), E(G))$ be our input graph, where every vertex has
 degree~$d$, and let $v^* \in V(G)$ be the distinguished vertex.  We
 denote by $q:=|V(G)|$ the number of vertices in~$G$ and by
 $r:=|E(G)|=qd/2$ the number of edges in~$G$.  In our construction, we
 use the following four types of candidates:
 \begin{enumerate}
 \item The special \emph{bar candidate}~$b$. We form the voters in
   such a way that $b$ belongs to every winning committee.
   
 \item For every vertex~$v \in V(G)$ we introduce a \emph{vertex
     candidate}~$c(v)$. The intention is that $c(v)$ belongs to some
   winning committee if and only if $v$ belongs to some minimum-size
   vertex cover. For a set $Y$ of vertices, we write $c(Y)$ to mean
   the set of corresponding vertex candidates.

 \item For every edge~$e \in E(G)$, we introduce a set $D(e)$ of
   $(2qr)^4$ edge-$e$ candidates. The intention is that these
   candidates never belong to a winning committee, but their presence
   ensures that a winning committee must include candidates
   corresponding to a vertex cover.

 \item We introduce~$q$ \emph{filler candidates} $f_1$, 
   $\dots$, $f_{q}$. The intention is that these candidates fill-in
   the places in the winning committee that are not taken by the
   vertex candidates, in such a way that the more filler candidates a
   committee includes, the lower is its dissatisfaction score.
 \end{enumerate}
 We set the committee size to be $q+1$. The main idea of the
 construction is that every winning committee has to contain the bar
 candidate, as few of the vertex candidates as possible (but so that
 they form a vertex cover), and arbitrary filler candidates to reach
 the committee size. Let~$c(v^*)$~be the distinguished candidate.
 
 \paragraph{Voters.}
 Following Remark~\ref{rem:dissat} we focus on the dissatisfaction
 score instead of the $\bordacc$-score of a committee.  A decisive
 construction property will be that the dissatisfaction score of a
 winning committee will be at most $X:=2qr$.  We form the following
 voter groups:
 \begin{enumerate}
 \item The \emph{bar group} contains $X+1$~\emph{bar voters}, each with
   preference order:
    \begin{align*}
      b \succ C \setminus \{b\}.
    \end{align*}
    That is, every bar voter prefers~$b$ over all other
    candidates.

  \item The \emph{edge group} contains two voters for each
    edge~$e=\{x,y\}$ with the following preference orders:
    \begin{align*}
      c(x) \succ c(y) \succ D(e) \succ b \succ \dots, \\ 
      c(y) \succ c(x) \succ {D(e)} \succ b \succ \dots, 
    \end{align*}
    where the candidates behind~$b$ are ranked arbitrarily.
  \item The \emph{filler group} contains $2r$ voters for each filler
    candidate~$f_i$. The voters associated with candidate $f_i$ have
    the following preference orders:
    \begin{align*}
      f_i \succ b \succ \dots,
    \end{align*}
    where the candidates behind~$b$ are ranked arbitrarily. Altogether, there
    are $X = 2qr$ voters in the filler group.
 \end{enumerate}
 This completes the construction. We see that it can be computed in
 polynomial time.

 \paragraph{Correctness.}
 Let us now analyze the properties of the constructed election.
 First, we note that every winning committee must contain candidate~$b$. In
 particular, if a committee does not contain~$b$, then its dissatisfaction score
 is at least $X+1$ due to the bar voters. Second, the
 committee~$\{b\} \cup c(V(G))$ has score $X$ (the bar voters and the edge
 voters provide dissatisfaction score $0$, and
 each of the $X$ filler voters provides dissatisfaction score
 $1$). Thus no committee with score greater than $X$ is winning (and
 this includes all the committees that do not include $b$).

 

 We are now ready to show the correctness of the reduction which is
 done via the following claim.

 \begin{claim} \label{thetap2-winning-committee-structure} Let $S$~be
   a $\bordacc$-winning committee for the above-described election.
   Then $S$~must be of the form~$\{b\} \cup c(V') \cup F'$, where
   $V'$~is a minimum size vertex cover for~$G$ and $F'$ is a set of
   $q-|V'|$~arbitrary filler candidates.  Moreover, $S$~has a
   dissatisfaction score of $2r + |V'|\cdot (2r-d)$.
 \end{claim}

   To prove the claim, let~$S$ be some $\bordacc$-winning committee. Let us consider some
   edge $e = \{v,u\}$; we note that $S$ does not contain any of the
   edge candidates from $D(e)$. On the one hand, if $S$ already
   contained contains~$c(v)$ or $c(u)$, then replacing one of the
   edge-$e$ candidates with some arbitrary filler candidate would give
   a committee with a smaller dissatisfaction score.  On the other
   hand, if $S$~did not contain either of~$c(v)$ or $c(u)$, then
   replacing an $e$-edge candidate with~$c(v)$ or with~$c(u)$ would
   give a committee with a smaller dissatisfaction score.

   Further, we note that $S$ must include some set~$c(V')$ of candidates, where
   $V'$ is a vertex cover of $G$. Otherwise there would
   be some edge $e = \{u,v\}$, whose associated voters would provide
   dissatisfaction score at least $(2qr)^4 > X$ ($c(u)$ and $c(v)$
   would not be in the committee because it did not contain a vertex
   cover and the edge candidates would not be included by the
   reasoning from the previous paragraph).

   As a consequence of the above reasoning (and of the fact that $b$
   belongs to every winning committee), we see that $S$ is of the form
   $\{b\} \cup c(V') \cup F'$, $V'$ is a vertex cover, and $F'$ is an
   arbitrary subset of $q-|V'|$ filler candidates (note that the
   dissatisfaction score of the committee depends on the number of the
   filler candidates, but not on their identities).  Let us now
   compute the dissatisfaction score of such an~$S$.

   First, there is no dissatisfaction from the voters in the bar
   group.  To see the dissatisfaction from the voters in the edge
   group, note that for each edge~$e$ the two corresponding voters
   either contribute dissatisfaction score~$1$ (when exactly one
   endpoint of~$e$ is in~$V'$) or they contribute dissatisfaction
   score~$0$ (when both endpoints of~$e$ are in~$V'$).  A vertex cover
   of size~$|V'|$ is incident to edges exactly~$|V'|\cdot d$
   times. Since a vertex cover is incident to each of the~$r$ edges at
   least once, it holds that it is incident to
   $|V'|\cdot d - r$~distinct edges exactly two times. Thus,
   $|V'|\cdot d - r$~distinct edges have both endpoints in~$V'$ and
   $r - (|V'|\cdot d - r) = 2r - |V'|\cdot d$ edges have only one
   endpoint in~$V'$. Thus the voters in the edge voter group
   contribute dissatisfaction score $2r - |V'|\cdot d$.  The voters in
   the filler group, by definition, contribute
   $(q-|F'|) \cdot 2r=|V'|\cdot 2r$ to the dissatisfaction score.  In
   total, the dissatisfaction score of our winning committee~$S$ is:
   \begin{align*} \label{dissat-formula}
    2r - |V'|\cdot d + |V'|\cdot2r \;=\; 2r + |V'|\cdot (2r - d).
   \end{align*}
   Based on this formula, we see that the vertex cover $V'$ induced by
   committee $S$ must have the smallest cardinality, because this
   leads to the lowest dissatisfaction score of $S$ (as $2r > d$).
   This completes the proof of the claim.
   

   By a reasoning analogous to that from the proof
   of~Claim~\ref{thetap2-winning-committee-structure}, we see that if $V'$ is a
   minimum-size vertex cover for $G$, then every committee of the form $\{b\}
   \cup c(V') \cup F'$, where $F'$ includes $q-|V'|$ arbitrary filler
   candidates, is winning in our election. This completes the proof
   of~Theorem~\ref{thm:bccmember}.

\end{document}